\documentclass[nohyperref]{article}

\usepackage{microtype}
\usepackage{graphicx}
\usepackage{subfigure}
\usepackage{booktabs} 

\usepackage{hyperref}

\usepackage[accepted]{icml2023}

\usepackage{amsmath}
\usepackage{amssymb}
\usepackage{mathtools}
\usepackage{amsthm}

\usepackage[capitalize,noabbrev]{cleveref}

\theoremstyle{plain}
\newtheorem{theorem}{Theorem}[section]
\newtheorem{proposition}[theorem]{Proposition}
\newtheorem{lemma}[theorem]{Lemma}
\newtheorem{corollary}[theorem]{Corollary}
\theoremstyle{definition}
\newtheorem{definition}[theorem]{Definition}

\theoremstyle{remark}
\newtheorem{remark}[theorem]{Remark}

\usepackage[textsize=tiny]{todonotes}

\icmltitlerunning{Abstracting Imperfect Information Away from Two-Player Zero-Sum Games}

\usepackage{enumitem}

\begin{document}

\twocolumn[
\icmltitle{Abstracting Imperfect Information Away from Two-Player Zero-Sum Games}



\icmlsetsymbol{equal}{*}
\icmlsetsymbol{atmeta}{\dag}

\begin{icmlauthorlist}
\icmlauthor{Samuel Sokota}{cmu,atmeta}
\icmlauthor{Ryan D'Orazio}{mila}
\icmlauthor{Chun Kai Ling}{cmu}
\icmlauthor{David J. Wu}{meta}
\icmlauthor{J. Zico Kolter}{cmu,bosch}
\icmlauthor{Noam Brown}{meta}
\end{icmlauthorlist}

\icmlaffiliation{cmu}{Carnegie Mellon University}
\icmlaffiliation{mila}{Mila, Université de Montréal}
\icmlaffiliation{meta}{Meta AI}
\icmlaffiliation{bosch}{Bosch Center for AI}

\icmlcorrespondingauthor{Samuel Sokota}{ssokota@andrew.cmu.edu}

\icmlkeywords{Machine Learning, ICML}

\vskip 0.3in
]



\printAffiliationsAndNotice{\atmeta}  


\begin{abstract}
In their seminal work, \citet{nayyar} showed that imperfect information can be abstracted away from common-payoff games by having players publicly announce their policies as they play.
This insight underpins sound solvers and decision-time planning algorithms for common-payoff games.
Unfortunately, a naive application of the same insight to two-player zero-sum games fails because Nash equilibria of the game with public policy announcements may not correspond to Nash equilibria of the original game.
As a consequence, existing sound decision-time planning algorithms require complicated additional mechanisms that have unappealing properties.
The main contribution of this work is showing that certain regularized equilibria do not possess the aforementioned non-correspondence problem---thus, computing them can be treated as perfect-information problems.
Because these regularized equilibria can be made arbitrarily close to Nash equilibria, our result opens the door to a new perspective to solving two-player zero-sum games and yields a simplified framework for decision-time planning in two-player zero-sum games, void of the unappealing properties that plague existing decision-time planning approaches.
\end{abstract}

\section{Introduction}

In single-agent settings, dynamic programming \citep{dp} is the bedrock for reinforcement learning \citep{sutton2018reinforcement}, justifying approximating optimal policies by backward induction and facilitating a simple framework for decision-time planning.
One might hope that dynamic programming could provide similar grounding in multi-agent settings for well-defined notions of optimality, like optimal joint policies in common-payoff games, Nash equilibria in two-player zero-sum (2p0s) games, and team correlated equilibria in two-team zero-sum (2t0s) games.
Unfortunately, this is not straightforwardly the case when there is imperfect information---a term that we use to refer to games in which one player has knowledge that another does not or two players act simultaneously.
This difficulty arises from two causes, which we call the \textit{backward dependence problem} and the \textit{non-correspondence problem}.

\begin{figure*}[h]
    \centering
    \includegraphics[width=.99\textwidth]{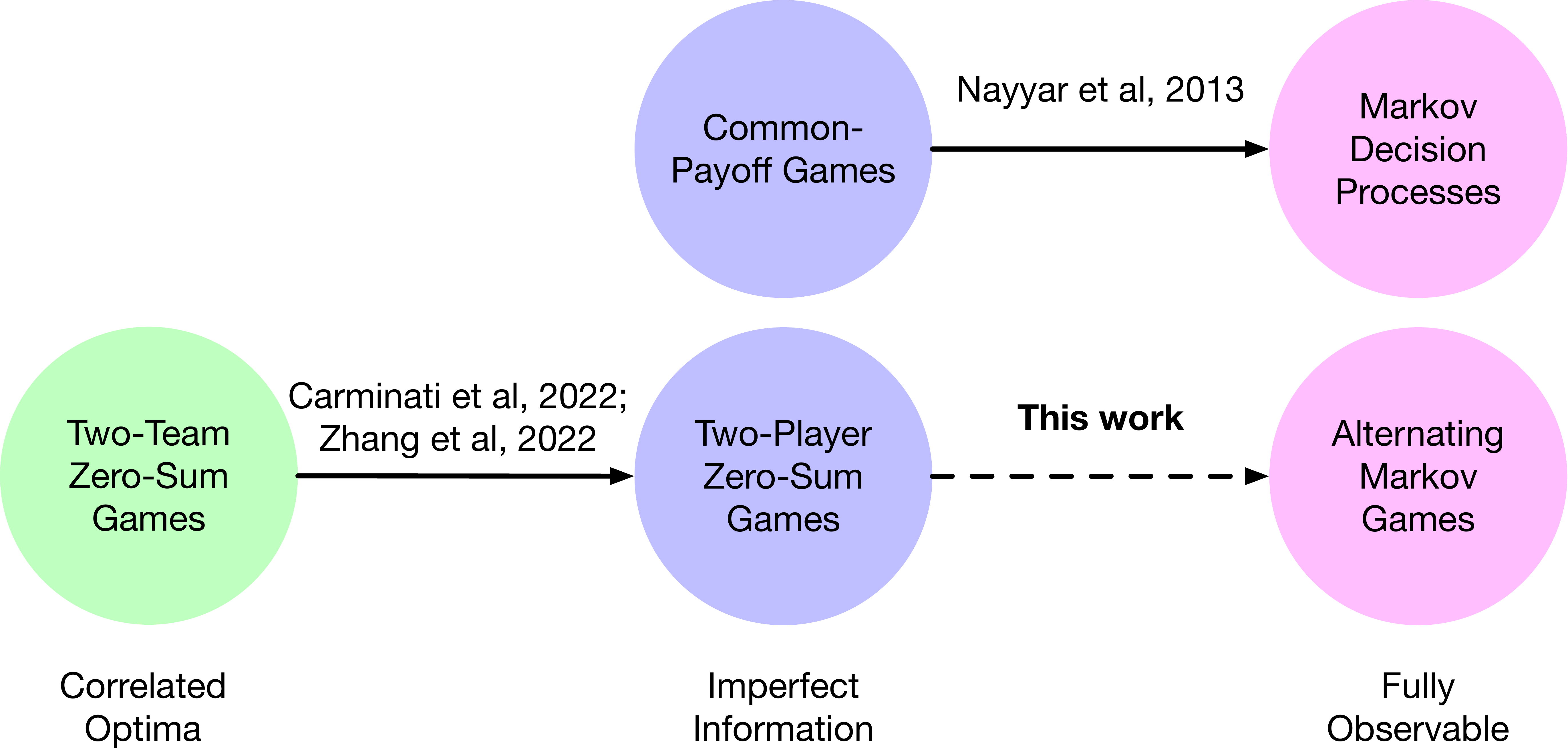}
    \caption{Our main contribution in the context of related work, at an abstract level. Solid lines denote reductions; the dashed line denotes a reduction that holds under a class of regularized objectives.}
    \label{fig:diagram}
\end{figure*}

\textit{The backward dependence problem} is that computing the expected return starting from a decision point generally requires knowledge about policies that were played up until now, in addition to the policies that will be played going forward.
This is in stark contrast to perfect information settings, in which the expected return starting from a decision point is independent of the policy played before arriving at the decision point.
As a result of this bidirectional temporal dependence, backward induction arguments that work in perfect information settings fail in imperfect information settings.

In their seminal work, \citet{nayyar} showed that the backward dependence problem can be resolved by having players publicly announce their policies as they play.
Using this insight, a common-payoff game can be transformed into a Markov decision process (MDP) that we call the public belief Markov decision processes (PuB-MDP).
Importantly, deterministic optimal policies in the PuB-MDP can be mapped back to optimal joint policies of the original common-payoff game.

Having players publicly announce their policies can also be used to transform 2p0s games into alternating Markov games (AMGs) with public belief states \citep{Wiggers2016StructureIT,nayyar-zs,brown2020rebel,buffet:hal-03080287,delage:hal-03523399,kartik2021upper}, which we call public belief alternating Markov games (PuB-AMGs).
AMGs are fully-observable turn-based games (like Go and chess)\footnote{Formally, the board states of Go and chess are technically not actually Markov because of rules like the threefold repetition rule.
Nevertheless, they capture the spirit of fully-observable turn-based games.} and, therefore, are amenable to dynamic programming-based approaches \citep{Littman96algorithmsfor}.
Unfortunately, computing Nash equilibria of PuB-AMGs carries little value because these Nash equilibria may not correspond with Nash equilibria in the original game \citep{cfrd,Ganzfried,brown2020rebel,sustr}. 
Indeed, as we will show, they may correspond with arbitrarily exploitable policies.
We call this problem \textit{the non-correspondence problem}.

The main contribution of this work is showing that regularized minimax objectives that guarantee unique equilibria in subgames do not suffer from the non-correspondence problem.
In other words, computing these uniqueness-guaranteeing equilibria can be reduced to computing the associated equilibria in the PuB-AMG.
Because uniqueness can be guaranteed using arbitrarily small amounts of entropy regularization \citep{fforel}, our reduction is straightforward to apply in practice and yields solutions that can be made arbitrarily close to Nash equilibria.

We highlight three points regarding this reduction:
\begin{enumerate}[leftmargin=*]
    \item It is the first reduction of its kind in literature; specifically, it is the first equilibrium preserving transformation from imperfect information 2p0s games to perfect information 2p0s games.
    \item It yields a simple framework for decision-time planning in 2p0s games with desirable continuity properties.
    In contrast, existing approaches \citep{libratus,safenested,deepstack,brown2020rebel,pog} are hampered by a number of complications that involve undesirable aspects, including discontinuous functions.
    \item It can be applied across the whole class of 2t0s games (as depicted by Figure \ref{fig:diagram}) because of the recent results of \citet{carminati1,Zhang2022TeamBD,carminati2}, who showed that 2t0s games can be cast as 2p0s games.
\end{enumerate}

\section{Notation}

We introduce two sets of formalisms.
The first, which we call
finite-horizon sequential games, describes settings in which players act one-at-a-time and in which the game terminates after a fixed number of steps.
This setting is equivalent to perfect recall timeable \citep{timeable} extensive form games---see, for example, \citet{fosg} for more details.

The second formalism, which we call finite-horizon fully-observable sequential games, captures a special case of the previous setting in which there is a Markov state that is observable to all players.
We use this formalism to express games with public policy announcements. 

\subsection{Finite-Horizon Sequential Games}
Symbolically, we say a setting is a finite-horizon sequential game if it can be described by a tuple
\[
\langle \mathbb{A},
[\mathbb{O}_i], \mathbb{O}_{\text{pub}}, [\mathbb{H}_i], \mathbb{H}_{\text{pub}},
\mathbb{H},
\mu,
[\mathcal{O}_i], \mathcal{O}_{\text{pub}},
[\mathcal{R}_i], \mathcal{T}, T \rangle,
\] 
where 
\begin{itemize}[leftmargin=*]
    \item $i$ ranges from $0$ to $N-1$ and $\iota$ denotes the acting player.\footnote{Our usage of $\iota$ is informal but unambiguous in context.}
    \item $\mathbb{A}$ is the set of actions.
    \item $\mathbb{O}_i$ is the set of private observations for player $i$.
    \item $\mathbb{O}_{\text{pub}}$ is the set of public observations (i.e., observations that are immediate common knowledge among players).
    \item $\mathbb{H}_i = \cup_{t=0}^{T-1} (\mathbb{O}_{\text{pub}} \times \mathbb{O}_i)^t \times \mathbb{O}_{\text{pub}} \times \mathbb{O}_i$ is player $i$'s action-observation histories (AOHs).
    \item $\mathbb{H}_{\text{pub}} = \cup_{t=0}^{T-1} \mathbb{O}_{\text{pub}}^t$ is the set of public histories.
    \item $\mathbb{H} \subset \mathbb{H}_0 \times \dots \times \mathbb{H}_{N-1}$ is the set of histories.
    \item $\mu \in \Delta(\mathbb{H}(h_{\text{pub}}^0))$ is the initial history distribution.
    \item $\mathcal{O}_i \colon \mathbb{H} \to \mathbb{O}_i$ is player $i$'s observation function.\footnote{We assume that, if $i$ acts a time $t$, $i$'s action is included in its private observation at time $t+1$}
    \item $\mathcal{O}_{\text{pub}} \colon \mathbb{H} \to \mathbb{O}_{\text{pub}}$ is the public observation function.
    \item $\mathcal{R}_i \colon \mathbb{H} \times \mathbb{A} \to \mathbb{R}$ is the player $i$'s reward function.
    \item $\mathcal{T} \colon \mathbb{H} \times \mathbb{A} \to \Delta(\mathbb{H})$ is the transition function. 
    \item $T$ is the time horizon at which the game terminates.
\end{itemize}
For a given $h_{\text{pub}} \in \mathbb{H}_{\text{pub}}$, we use $\mathbb{H}_i(h_{\text{pub}})$ to denote the set of AOHs for player $i$ that are consistent with $h_{\text{pub}}$ and $\mathbb{H}(h_{\text{pub}})$ to denote the set of histories that are consistent with $h_{\text{pub}}$.
Also, for a history $h$, we use $h_{\iota}$ to denote the AOH for the acting player at history $h$.
We use capitals of the same letters to denote random variables of the same types.
We use $\pi \colon \cup_i \mathbb{H}_i \to \Delta (\mathbb{A})$ to denote the joint policy of the players and $\pi_i \colon \mathbb{H}_i \to \Delta (\mathbb{A})$ to denote player $i$'s policy.
We use $-i$ to denote ``all players except player $i$''.

\paragraph{Special Cases}
This work will make use of the following special cases.
\begin{itemize}[leftmargin=*]
    \item \textbf{Two team zero sum}: Games in which $\{0, \dots, N-1\}$ is a disjoint union of two blocks, where $\forall i, j, \mathcal{R}_i = \mathcal{R}_j$ if $i, j$ belong to the same block and $\mathcal{R}_i = -\mathcal{R}_j$ if $i, j$ belong to opposite blocks.
    \item \textbf{Common payoff}: Games in which $\forall i, j, \mathcal{R}_i = \mathcal{R}_j$.
    \item \textbf{Two player zero sum}: Games in which $N=2$ and $\mathcal{R}_0 = - \mathcal{R}_1$.
\end{itemize}
In these special cases, the reward of all players is uniquely determined by the reward of any individual player.
Thus, we will drop the player index $i$ on the reward function and use $\mathcal{R}=\mathcal{R}_0$.

\paragraph{Subgames} For a given finite-horizon sequential game, we use the term subgame to refer to a game that begins with initial history distribution $\mu \in \Delta(\mathbb{H}(h_{\text{pub}}))$ for some particular $h_{\text{pub}}$ reflecting public information revealed so far and is otherwise the same as the original game.

\subsection{Finite-Horizon Fully-Observable Sequential Games} We use the terminology finite-horizon fully-observable sequential game to describe tuples 
\[
\langle \mathbb{A},
\mathbb{S}, s^0, [\mathcal{R}_i], \mathcal{T}, T \rangle,
\]
where
\begin{itemize}[leftmargin=*]
    \item $\mathbb{S}$ is the set of states.
    \item $s^0 \in \mathbb{S}$ is the initial state.
    \item $\mathcal{R}_i \colon \mathbb{S} \times \mathbb{A} \to \mathbb{R}$ is the player $i$'s reward function.
    \item $\mathcal{T} \colon \mathbb{S} \times \mathbb{A} \to \Delta(\mathbb{S})$ is the transition function.
    \item $i$, $\iota$, $\mathbb{A}$, and $T$ are defined as they were in the finite-horizon sequential game formalism.
\end{itemize}

We use $\pi \colon \mathbb{S} \to \Delta (\mathbb{A})$ to denote the joint policy and $\pi_i$ to denote player $i$'s policy.

\paragraph{Special Cases} In the fully-observable context, we are interested in the following settings.
\begin{itemize}[leftmargin=*]
    \item \textbf{Markov decision processes (MDPs)}: Games in which $N=1$.
    \item \textbf{Alternating Markov games (AMGs)}: Games that are two player zero sum.
\end{itemize}
As before, we will use  $\mathcal{R}=\mathcal{R}_0$ for conciseness.

\paragraph{Subgames} For a given finite-horizon fully-observable sequential game, we use the term subgame to refer to a game that begins with initial state $s^0=s$ for some particular $s$ and otherwise proceeds by the same rules of the original game.

\section{Background}

\paragraph{The Backward Dependence Problem} To illustrate the presence of the backward dependence problem in even very simple settings, we show a cooperative matching pennies game in Figure \ref{fig:mp}.
The goal of the game is for the blue player and the red player to select the same side of a coin.
The blue player moves first; then, without observing the blue player's choice, the red player moves second.
Because the red player does not observe the blue player's choice (as denoted by the dotted line between the two nodes), it must make the same decision at both nodes.

Now, let us consider the value for the red player.
In perfect-information settings, because the red player is at a penultimate node, such a value would be equal to the expected return for the best action, independent of any prior events.
However, here, with imperfect information, the expected return for the best action is equal to $\max(p, 1-p)$ where $p$ is the probability that the blue player selects \textit{heads}.
If the blue player's policy is unknown, there is no way to compute this value, illustrating that the backward induction approach to learning in perfect information settings fails in imperfect information settings.

\begin{figure}
  \begin{center}
    \includegraphics[width=0.48\textwidth]{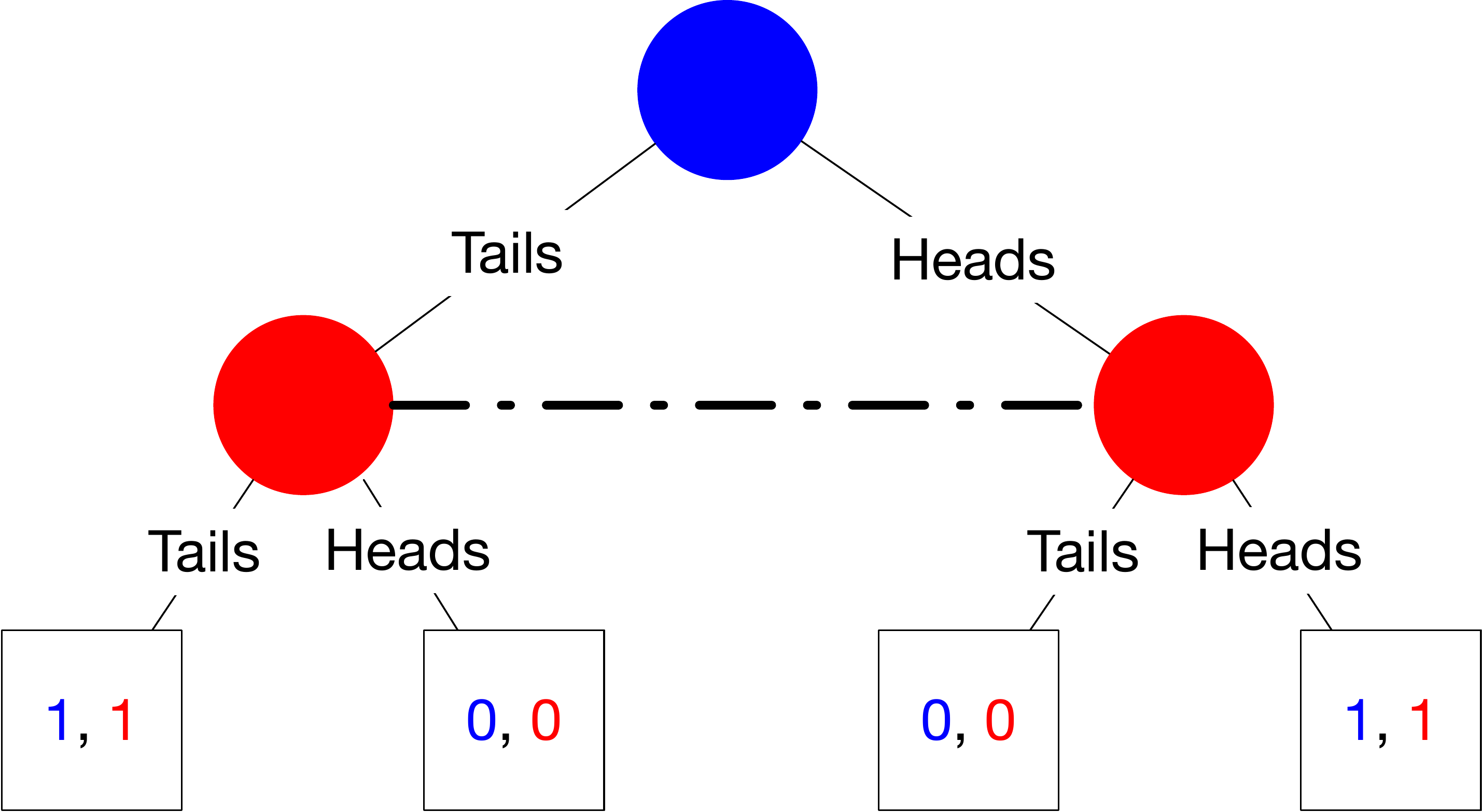}
  \end{center}
  \caption{The best action for the red player depends on the blue player's policy.}
  \label{fig:mp}
\end{figure}

\paragraph{The Public Belief Markov Decision Process} In their seminal work, \citet{nayyar} described a reduction for turning common-payoff games into partially observable Markov decision processes (POMDPs) in such a way that circumvents the backward dependence problem.
This reduction can be chained with the well-known belief-state reduction from POMDPs to MDPs to construct public belief state MDPs (PuB-MDPs).
We describe the composition of these reductions.

Let
\[
\langle \mathbb{A},
[\mathbb{O}_i], \mathbb{O}_{\text{pub}}, [\mathbb{H}_i], \mathbb{H}_{\text{pub}},
\mathbb{H},
\mu,
[\mathcal{O}_i], \mathcal{O}_{\text{pub}},
\mathcal{R}, \mathcal{T}, T \rangle,
\] 
be a finite-horizon common-payoff game.
Then we define the associated PuB-MDP as the following finite-horizon fully-observable sequential game
\[
\langle \tilde{\mathbb{A}},
\tilde{\mathbb{S}},
\tilde{s}^0,
 \tilde{\mathcal{R}}, \tilde{\mathcal{T}}, \tilde{T} \rangle,
\]
where
\begin{itemize}[leftmargin=*]
    \item $i = \iota = 0$.
    \item $\mathbb{\tilde{A}} = \{\tilde{a} \mid \tilde{a}  \colon \mathbb{H}_{\iota}(h_{\text{pub}}) \to \mathbb{A}, h_{\text{pub}} \in \mathbb{H}_{\text{pub}}\}$ is the set of \textit{prescriptions}.
    \item $\mathbb{\tilde{S}} = \cup_{h_{\text{pub}}} \Delta(\mathbb{H}(h_{\text{pub}})) $ is the set of public belief states (PBSs).
    \item $\tilde{s}^0 = \mu$ is the initial PBS.
    \item $\mathcal{\tilde{R}} \colon \tilde{s}, \tilde{a} \mapsto \mathbb{E}_{H \sim \tilde{s}} \mathcal{R}(H, \tilde{a}(H_{\iota}))$.
    \item $\mathcal{\tilde{T}}(\tilde{s}_{o^{t+1}_{\text{pub}}} \mid \tilde{s}^t, \tilde{a}) = 
    \mathbb{E}_{H^t \sim \tilde{s}^t} \mathcal{P}(o_{\text{pub}}^{t+1} \mid H^t, \tilde{a}(H_{\iota}^t))$\\
    where the PBS $\tilde{s}_{o^{t+1}_{\text{pub}}}$ is defined by\\
    $\tilde{s}_{o^{t+1}_{\text{pub}}}(h^{t+1}) {=} \mathbb{E}_{H^t \sim \tilde{s}^t} \mathcal{P}(h^{t+1} \mid H^t, \tilde{a}(H_{\iota}^t), o_{\text{pub}}^{t+1})$
    \item $\tilde{T} = T$.
\end{itemize}

\citet{nayyar} showed that optimal deterministic policies in the PuB-MDP correspond with optimal joint policies for the common payoff game.
Indeed, for the matching pennies game described in Figure \ref{fig:mp}, we can see that the PuB-MDP perspective resolves the backward dependence problem because the red player observes the blue player's prescription.
If the blue player's prescription maps to heads, the red player can determine that playing heads has a value of 1 whereas playing tails has a value of 0 (and vice versa if the blue player's prescriptions maps to tails).
Thus, the players can arrive at an optimal joint policy of the original game.

For a more detailed discussion on the PuB-MDP, see, e.g.,  \citep{capi_thesis,capi_paper}.

\section{The Public Belief Alternating Markov Game}

It is also possible to map 2p0s games to symmetric-information 2p0s games using policy announcements \citep{nayyar-zs,kartik2021upper}.\footnote{A symmetric-information game is one in which all players receive identical observations.}
This mapping can be chained with a belief-state transformation to construct games that we call public belief AMGs (PuB-AMGs) \citep{Wiggers2016StructureIT,nayyar-zs,brown2020rebel,buffet:hal-03080287,delage:hal-03523399,kartik2021upper}.\footnote{We note that details of the models studied in existing works sometimes differ from one another or from that of this work.}
In the main body below, we describe the composition of these reductions; we provide a brief discussion on these reductions as separate entities in Section \ref{sec:reduction} of the appendix.

Let
\[
\langle \mathbb{A},
[\mathbb{O}_i], \mathbb{O}_{\text{pub}}, [\mathbb{H}_i], \mathbb{H}_{\text{pub}},
\mathbb{H},
\mu,
[\mathcal{O}_i], \mathcal{O}_{\text{pub}},
[\mathcal{R}_i], \mathcal{T}, T \rangle,
\] 
be a finite-horizon 2p0s sequential game.
Then we define the associated PuB-AMG as the following finite-horizon fully-observable sequential game
\[
\langle \tilde{\mathbb{A}},
\tilde{\mathbb{S}},
\tilde{s}^0,
 \tilde{\mathcal{R}}, \tilde{\mathcal{T}}, \tilde{T} \rangle,
\]
where
\begin{itemize}[leftmargin=*]
    \item $i$ ranges from $0$ to $1$ and $\iota  \in \{0, 1\}$ is the acting player.
    \item $\mathbb{\tilde{A}} = \{\tilde{a} \mid \tilde{a} \colon \mathbb{H}_{\iota}(h_{\text{pub}}) \to \Delta(\mathbb{A}), h_{\text{pub}} \in \mathbb{H}_{\text{pub}}\}$ is the set of \textit{public decision rules} (or decision rules, for short).
    Note that public decision rules differ from prescriptions in that they map to a distributions over actions, rather than actions.
    \item $\mathbb{\tilde{S}} = \cup_{h_{\text{pub}}} \Delta(\mathbb{H}(h_{\text{pub}})) $ is the set of public belief states (PBSs).
    \item $\tilde{s}^0 = \mu$ is the initial PBS.
    \item $\mathcal{\tilde{R}} \colon \tilde{s}, \tilde{a} \mapsto \mathbb{E}_{H \sim \tilde{s}} \mathbb{E}_{A \sim \tilde{a}(H_{\iota})} \mathcal{R}(H, A)$.
    \item $\mathcal{\tilde{T}}(\tilde{s}_{o^{t+1}_{\text{pub}}} \mid \tilde{s}^t, \tilde{a}) = 
    \mathbb{E}_{H^t \sim \tilde{s}^t}  \mathbb{E}_{A^t \sim \tilde{a}(H^t_{\iota})} \mathcal{P}(o_{\text{pub}}^{t+1} \mid H^t, A^t)$\\
    where the PBS $\tilde{s}_{o^{t+1}_{\text{pub}}}$ is defined by\\
    $\tilde{s}_{o^{t+1}_{\text{pub}}}(h^{t+1}) {=} \mathbb{E}_{H^t \sim \tilde{s}^t} \mathbb{E}_{A^t \sim \tilde{a}(H^t_{\iota})} \mathcal{P}(h^{t+1} {\mid} H^t, A^t, o_{\text{pub}}^{t+1})$.
    \item $\tilde{T}=T$.
\end{itemize}

Notice that the PuB-AMG closely resembles the PuB-MDP in structure, differing only in the number of players, the structure of the actions, and that an additional expectation is required in the reward and transition functions.

\paragraph{The Correspondence Mapping} As mentioned earlier, a Nash equilibrium in the PuB-AMG may be undesirable because it does not necessarily correspond to a Nash equilibrium in the original game.
Here, we make this notion precise by defining a correspondence function $\Pi^{\downarrow}$ that maps public belief joint policies to joint policies of the original game.
Given a PuB-AMG joint policy $\tilde{\pi}$, $\Pi^{\downarrow}(\tilde{\pi})$ is the joint policy that, for each AOH $h_{\iota}$, plays actions with the probability that $\tilde{\pi}$ would at $h_{\iota}$, assuming that $h_{\iota}$ was reached using $\tilde{\pi}$.
(See Section \ref{sec:def} for a more rigorous definition.)
\emph{Importantly, $\Pi^{\downarrow}(\tilde{\pi})_i$ can be implemented in practice by running $\tilde{\pi}_i$ under the assumption that the opponent is playing according to~$\tilde{\pi}_{-i}$.}

\begin{figure}
  \begin{center}
    \includegraphics[width=0.48\textwidth]{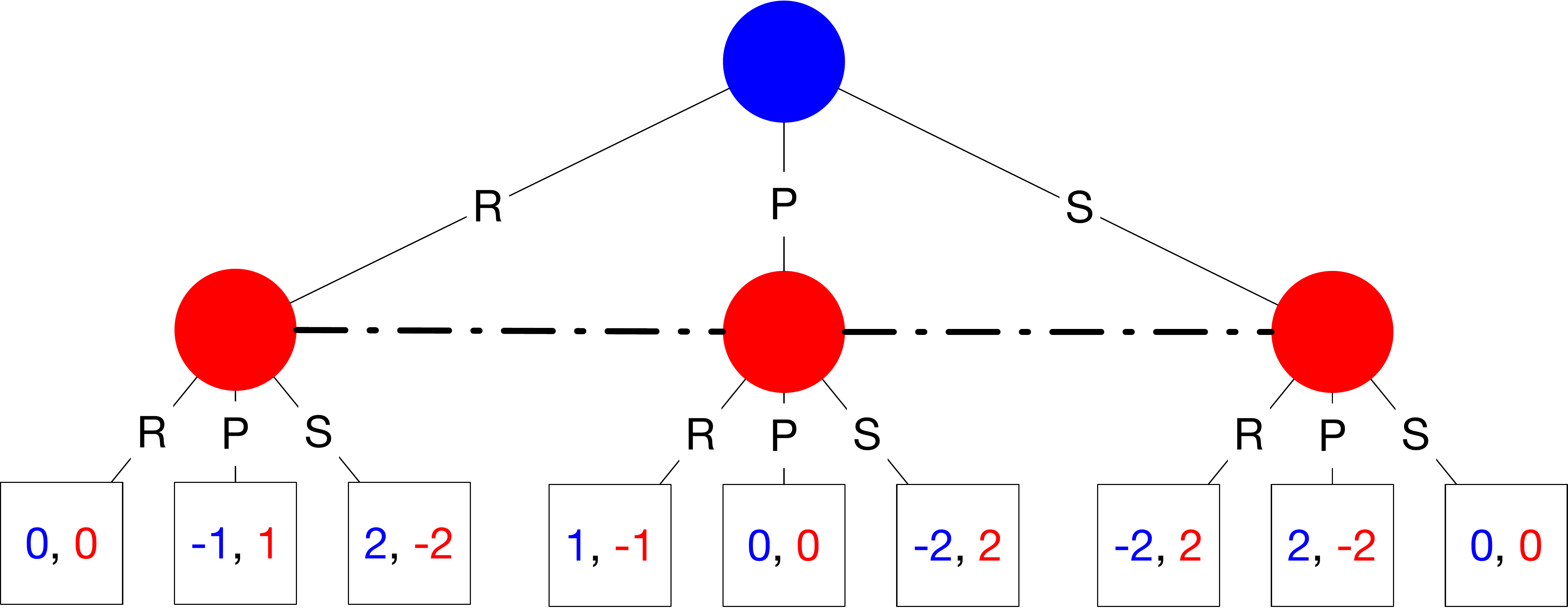}
  \end{center}
  \caption{A perturbed variant of rock-paper-scissors.}
  \label{fig:pert_rps_game}
\end{figure}

\paragraph{The Non-Correspondence Problem} We can now discuss non correspondence more rigorously.
To illustrate, we show the perturbed variant of rock-paper-scissors described in \citep{brown2020rebel} in Figure~\ref{fig:pert_rps_game}.
The game is perturbed in the sense that the payouts are doubled if either player plays scissors.
The unique Nash equilibrium of the game is $(R, P, S) \mapsto (0.4, 0.4, 0.2)$.

Similarly to before, the red player can compute the associated value for each of the blue player's decision rules in the PuB-AMG.
Thus, the blue player can determine that the Nash equilibrium policy maximizes its value.
It is at this point that the non-correspondence problem becomes apparent.
Because the red player is conditioning on the blue player's decision rule, it achieves the optimal value by playing any best response to the blue player.
In the perturbed rock-paper-scissors game, all policies are best responses to the Nash equilibrium.
Thus, there is nothing constraining the red player to the Nash equilibrium policy of the original game.

A similar argument, detailed in Section \ref{sec:worst_case}, leads to the following disappointing result.
\begin{proposition} \label{prop:worst_case}
 A PuB-AMG Nash equilibrium $\tilde{\pi}$ may correspond with a joint policy $\Pi^{\downarrow}(\tilde{\pi})$ that is maximally exploitable.
\end{proposition}

At an intuitive level, the non-correspondence problem arises because there is an important distinction between the public belief game and the original game.
Specifically, in the public belief game, players acting earlier are forced to reveal their decision rules to players acting later.
As a result, later acting players are able to ``slack off'' without losing any value because the earlier acting players cannot deviate to punish them.
In common-payoff games, this is a non issue because the interests of every player are aligned.
However, in 2p0s games, where there are adversarial interests, this distinction changes the strategic nature of the game in a more fundamental sense.

\section{Uniqueness-Guaranteeing Objectives}

\looseness=-1
To address the non-correspondence problem discussed in the previous section, we introduce a class of objectives that we call uniqueness guaranteeing (UG).
UG objectives are a kind of regularized objective, of the form defined below, that generalize the expected return objective in that it includes objectives that may have dependence on policies beyond the actions they select.

\begin{definition} \label{def:objective}
We use the term regularized minimax objective (or objective for short) to refer to mappings of the form
\[
\mathfrak{J} \colon \pi_0, \pi_1 \mapsto \mathbb{E} \left[\sum_{t=0}^{T-1} \mathfrak{R}(H^t, A^t, \pi(H^t_{\iota})) \mid \pi_0, \pi_1\right]
\]
where $\mathfrak{R}$ is a real-valued function.
Every regularized minimax objectives $\mathfrak{J}$  possesses a PuB-AMG analog $\mathfrak{\tilde{J}}$ that is equivalent to $\mathfrak{J}$ in the sense that $\mathfrak{\tilde{J}}(\tilde{\pi}) = \mathfrak{J}(\Pi^{\downarrow}(\tilde{\pi}))$ that is defined by
\[\mathfrak{\tilde{J}} \colon \tilde{\pi}_0, \tilde{\pi}_1 \mapsto \mathbb{E} \left[\sum_{t=0}^{T-1} \mathfrak{\tilde{R}}(\tilde{S}^t, \tilde{A}^t) \mid \tilde{\pi}_0, \tilde{\pi}_1\right],\]
where
\[\mathfrak{\tilde{R}} \colon (\tilde{s}, \tilde{a}) \mapsto \mathbb{E}_{H \sim \tilde{s}} \mathbb{E}_{A \sim \tilde{a}(H_{\iota})} \mathfrak{R}(H, A, \tilde{a}(H_{\iota})).\]
\end{definition}

UG objectives are regularized minimax objectives that are guaranteed to produce unique equilibria.
\begin{definition} \label{def:unique}
For a particular game, we say a minimax objective $\mathfrak{J}$ is UG if
\[
\max_{\pi_0} \min_{\pi_1} \mathfrak{J}(\pi_0, \pi_1)
\]
is guaranteed to have a unique solution $\pi_{\ast}$ for every subgame of that game.
\end{definition}

\subsection{Correspondence of Uniqueness-Guaranteeing Equilibria}
We can now state our main result---that the non-correspondence problem does not exist for equilibria induced by UG objectives.

\begin{theorem} \label{thm:correspondence}
If $\tilde{\pi}$ is an equilibrium of a PuB-AMG under a UG objective, then its corresponding joint policy $\Pi^{\downarrow}(\tilde{\pi})$ is the equilibrium in the original game under the same UG objective.
\end{theorem}

\begin{proof}
(Sketch) The first decision rule of any PuB-AMG equilibrium must correspond to the first decision rule of an equilibrium of the original game.
Furthermore, if the objective is UG, subgame equilibria must be restrictions of the equilibrium of the whole game.
Thus, by forward induction, PuB-AMG equilibria must correspond to the equilibrium of the original game.
\end{proof}
We detail the proof for Theorem \ref{thm:correspondence} in Section \ref{sec:cor}.
Due to recent work, Theorem \ref{thm:correspondence} can be generalized to the entire class of 2t0s games.
\begin{corollary} \label{cor:2t02}
Computing team-correlated equilibria of 2t0s games under UG objectives can be reduced to computing an equilibrium of a PuB-AMG under the same UG objectives.
\end{corollary}

This follows from combining the results of \citet{carminati1,Zhang2022TeamBD,carminati2}, who provide a reduction from 2t0s games to 2p0s games via intra-team public policy announcements, with Theorem \ref{thm:correspondence}.

\subsection{Continuity in the PuB-AMG with Uniqueness-Guaranteeing Objectives}

Theorem \ref{thm:correspondence} shows that UG equilibria do not suffer from the non-correspondence problem, meaning that computing UG equilibria in the PuB-AMG induces equilibria in the original game.
Here, we show that these equilibria are also continuous, a desirable condition for amenability to function approximation, under the mild assumption that $\mathfrak{R}$ is continuous.

\begin{definition}
In a perfect information game, a subgame perfect equilibrium is an equilibrium whose restriction to any subgame is an equilibrium of that subgame.
\end{definition}

\begin{definition}
The PuB-AMG subgame perfect equilibrium value function is defined by \[\tilde{v}_{\ast} \colon \tilde{s}^t \mapsto  \max_{\tilde{\pi}_0} \min_{\tilde{\pi}_1} \mathbb{E}
 \left[\sum_{t'=t}^{T-1} \mathfrak{\tilde{R}}(\tilde{S}^{t'}, \tilde{A}^{t'}) \mid \tilde{\pi}_{0}, \tilde{\pi}_1, \tilde{S}^t=\tilde{s}^t \right].\]
\end{definition}

\begin{theorem} \label{thm:val_cont}
Let $\mathfrak{J}$ guarantee the existence of an equilibrium in all subgames. Then the PuB-AMG subgame perfect equilibrium value function is a continuous function from the space of PBSs to real values.
\end{theorem}

Note that Theorem \ref{thm:val_cont} holds even if $\mathfrak{J}$ is not UG.

\begin{theorem} \label{thm:pol_cont}
Let $\mathfrak{J}$ be a UG objective be induced by a continuous $\mathfrak{R}$.
Then the PuB-AMG subgame perfect equilibrium induced by $\mathfrak{J}$ is a continuous function from the space of PBSs to the space of public decision rules.
\end{theorem}

The proofs for Theorem \ref{thm:val_cont} and Theorem \ref{thm:pol_cont} are detailed in Section \ref{sec:cont}.

\subsection{Sufficient Conditions for Uniqueness Guaranteeing}

In aggregate, the previous two sections show that UG PuB-AMG equilibria are continuous under mild assumptions and correspond with UG equilibria in the original game.
While these results are positive, it is important to realize that the relevancy of these results hinges on the existence of UG objectives with desirable solutions.
Fortunately, as we discuss below, such objectives exist.

\begin{definition}
We call the objective $\mathfrak{J}$ induced by 
\[
\mathfrak{R} \colon (h, a, \delta) \mapsto 
\begin{cases}
\mathcal{R}(h, a) - \alpha \text{KL}(\delta, \rho(h_{\iota})) & \iota = 0\\
\mathcal{R}(h, a) + \alpha \text{KL}(\delta, \rho(h_{\iota})) & \iota = 1,
\end{cases}
\]
for some reference policy $\rho$, a \textit{MiniMaxKL} objective.
\end{definition}

\begin{definition}
We call the objective $\mathfrak{J}$ induced by
\[
\mathfrak{R} \colon (h, a, \delta) \mapsto 
\begin{cases}
\mathcal{R}(h, a) + \alpha \mathcal{H}(\delta) & \iota = 0\\
\mathcal{R}(h, a) - \alpha \mathcal{H}(\delta) & \iota = 1,\\
\end{cases}
\]
a \textit{MiniMaxEnt} objective.
\end{definition}

\begin{remark}
MiniMaxEnt is the special case of MiniMaxKL in which $\rho$ is uniform.
\end{remark}

To our knowledge, MiniMaxKL objectives were introduced by \citet{fforel}, who showed the following result.

\begin{theorem}[\citet{fforel}] \label{thm:perolat}
 MiniMaxKL objectives are UG for interior~$\rho$ for any $\alpha > 0$.
\end{theorem}

Importantly, beyond being UG, MiniMaxKL objectives can achieve arbitrarily small exploitabilities, as is formalized by the proposition below. 
In aggregate, these results mean that it is possible to compute policies with arbitrarily small exploitabilities via computing the equilibria of MiniMaxKL objectives in the PuB-AMG.

\begin{proposition} \label{prop:close}
Let $\mathfrak{J}$ be a MiniMaxKL objective parameterized by a reference policy $\rho$, placing at least $\epsilon$ probability on every action, and regularization parameter $\alpha$.
Then the exploitability of the MiniMaxKL equilibrium is bounded by $\alpha T |\log \epsilon|$, where $T$ is the horizon of the game.
\end{proposition}

The proof of Proposition \ref{prop:close} is detailed in Section \ref{sec:suff}.

While the MiniMaxKL equilibrium is likely the most useful UG equilibrium concept, it is conceivable that other UG concepts may be useful.
Thus, we provide a generalization to a larger class of regularized objectives in Theorem \ref{thm:suff} in Section \ref{sec:suff}.
\section{Discussion}

\paragraph{Use Cases} There are at least three main ways to approach solving regularized PuB-AMGs. The first is to adapt heuristic search value iteration \citep{hsvi} into a tabular regularized PuB-AMG solver.
Encouragingly, this has already been done for PuB-MDPs \citep{Dibangoye} and for unregularized Pub-AMGs \citep{valuedp,buffet:hal-03080287,delage:hal-03523399}.

The second is to use the regularized PuB-AMG as a building block for model-free deep reinforcement learning agents.
This approach would look similar to BAD \cite{bad}, which is a policy gradient method that was applied to an approximate PuB-MDP in Hanabi \cite{BARD2020103216}.
We believe that it is possible that a BAD-like approach in regularized PuB-AMG would be better suited to a game like poker, where it is convenient to tabularly track the PBS, than it was to Hanabi, where \citet{bad} required complicated posterior approximation techniques.

The third is to use the regularized PuB-AMG as a building block for expert iteration \citep{exit,exit-thesis} with function approximation.
This approach would look almost identical to ReBeL \citep{brown2020rebel} but have a few key differences:
i) It would use a regularized objective, rather than an unregularized one as ReBeL does; ii) It would use the beliefs induced by its own policy at test time, rather than the fictitious beliefs that ReBeL uses;
iii) It would (optionally) be able to perform re-planning (e.g., wherein a multi-ply search is only used to make the immediate decision), whereas ReBeL must play its search policy until the end of the subgame that was searched over, iv) It would (optionally) be able to perform additional search iterations at test-time, whereas ReBeL is required to use the same number of search iterations as it did during training.

\paragraph{On the Role of Regularization} One possible set of concerns regarding these proposed use cases is that: i) to achieve good performance in these use cases, it may be necessary to approximate equilibria of objectives having small amounts of regularization; ii) approximating equilibria with small amounts of regularization may be too difficult.
In tabular settings, i) may be true if the goal is to achieve competitive performance with methods not based on regularization, such as counterfactual regret minimization (CFR) \citep{cfr}; however, \citet{mmd} recently showed that regularization-based methods can be made competitive with CFR in tabular settings by slowly annealing the amount of regularization, suggesting that ii) may be false.
On the other hand, in larger settings in which function approximation is necessary, ii) may be true; however, \citet{mmd} also showed that deep reinforcement learning approaches with substantial amounts of regularization can achieve good performance in terms of approximate exploitability, suggesting that i) may be false.

\section{Experiments}

We perform two experiments in which we naively tabularly solve small PuB-AMGs under regularized objectives using magnetic mirror descent \citep{mmd} to offer further evidence for our results.
We show the results for perturbed rock-paper-scissors Figure \ref{fig:rps} and include results for Kuhn poker, as well as the details of our solving procedures, in Section \ref{sec:exp}.

\begin{figure*}[h]
    \centering
    \includegraphics[width=\textwidth]{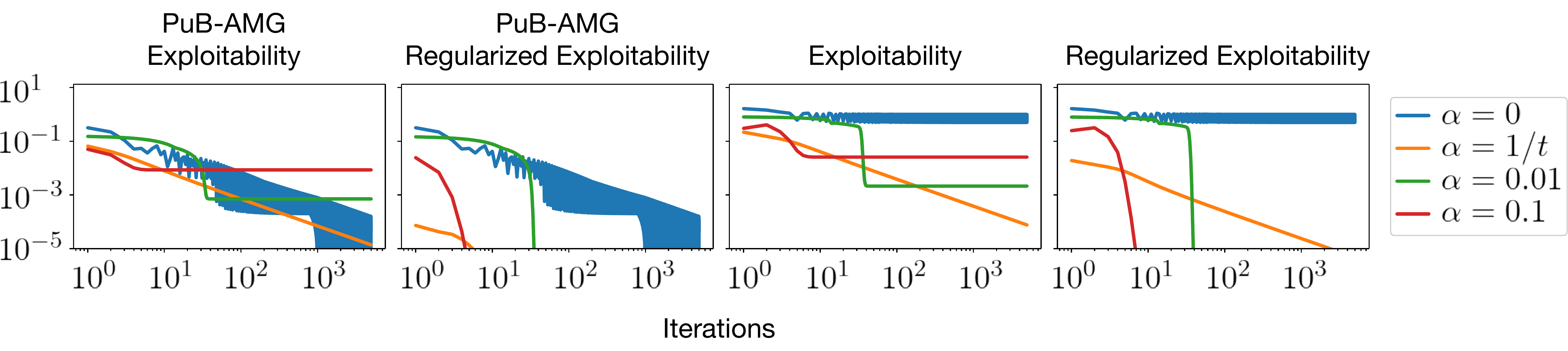}
    \caption{Results for perturbed rock-paper-scissors.}
    \label{fig:rps}
\end{figure*}

On the far left, we show exploitability in the PuB-AMG.
The iterates of the unregularized objective (blue) trend and the iterates of the objective with annealed regularization (orange) both trend toward zero.
The iterates of the objective with constant regularization converge to a constant positive exploitability.

On the middle left, we show the regularized exploitability (i.e., exploitability under the regularized objective) in the PuB-AMG of the objective associated for the iterate.
We observe that all objectives induce iterates that converge to zero, as intended.

On the middle right, we show the exploitability in the original game.
Because the non-correspondence problem exists for the second-moving player, the exploitabilities of the iterates from the unregularized objective (blue) remain high, despite that exploitability is going to zero in the PuB-AMG.
The objectives with fixed regularization (green, red) induce iterates that converge to lower, but non-zero, exploitability values.
The objective with annealed regularization (orange) induces iterates that converge to zero in exploitability.

\looseness=-1
On the far right, we show the regularized exploitability in the original game.
We observe that, as expected, the approaches with non-zero regularization that converge to zero regularized exploitability in the PuB-AMG also converge to zero regularized exploitability in the original game.
In contrast, the unregularized approach does not converge, despite converging in the PuB-AMG.

\section{Related Work}

\paragraph{Public Belief States in Common-Payoff Games} In the sense of providing reductions for multi-agent problems using PBSs, our work is similar to those of \citet{nayyar}, \citet{Dibangoye}, and \citet{Oliehoek2013SufficientPS}.
As discussed in the background, \citet{nayyar} provided a reduction from solving common-payoff games to solving belief MDPs; independently, \citet{Dibangoye} and \citet{Oliehoek2013SufficientPS} discovered similar reductions.
These ideas have been leveraged in a large body of work in decentralized control literature \citep{control2013,control2014a,control2014b,control2015,control2016a,control2016b,control2018a,control2018b,control2018c,control2019,control2021} and machine learning literature \citep{ml2013a,ml2013b,ml2014b,ml2014a,ml2018,bad,sparta,capi_paper,fickinger2021scalable,bft,ml2022}.
Use cases include game solving \citep{Dibangoye} and decision-time planning \citep{sparta,fickinger2021scalable,bft}.

\paragraph{Public Belief States in Two-Player Zero-Sum Games} PBSs have also been used in many works in the context of 2p0s games.
For our purposes, we taxonomize these into those concerned with studying the PuB-AMG \citep{Wiggers2016StructureIT,nayyar-zs,valuedp,buffet:hal-03080287,delage:hal-03523399,kartik2021upper} and those concerned with sound decision-time planning and expert iteration \citep{cfrd,margin,libratus,safenested,deepstack,multivalued,Zarick2020UnlockingTP,brown2020rebel,pog}.

Most of the former group is concerned with analyzing the structure of the PuB-AMG and using HSVI \citep{hsvi} to compute the equilibrium value of the game.\footnote{In concurrent work, \citet{delage_2022} show how an $\epsilon$-Nash equilibrium of the original game can be extracted from a variant of this approach without requiring UG objectives.}
Our work is complementary in the sense that it shows that solving a regularized PuB-AMG would yield a regularized equilibrium in the original game.

The latter group can be broken down into two subgroups, those that use opt-out values to circumvent the non-correspondence problem \citep{libratus,safenested,deepstack,pog} and that which uses no-regret learning to circumvent the non-correspondence problem \citep{brown2020rebel}.
Both possess substantial downsides.
For the opt-out value approach: i) the policy and value are discontinuous functions of the opt-out values\footnote{Though \citet{pog} show that certain approximate value functions can be made continuous.}, and ii) the opt-values must be approximated separately from self play.
For the no-regret learning approach: i) the search policy must be played for the entire subgame that was searched over (i.e., re-planning is not allowed), ii) the search algorithm must be no regret, iii) the policy is a discontinuous function of the PBS, and iv) the same number of search iterations must be used at test time as were used during training.
In contrast, decision-time planning using a regularized objective in the PuB-AMG involves no opt-out values, involves no discontinuities, allows for re-planning, is search-algorithm agnostic, and can use an arbitrary number of search iterations at test time.

\paragraph{MiniMaxKL Objectives in Two-Player Zero-Sum Games} A number of recent prior works have made use of MiniMaxEnt and MiniMaxKL objectives 
for the purpose of inducing last iterate convergence \citep{fforel,Cen2021FastPE,zeng2021,mmd,stratego}.
While we also make use of these objectives, our use case (eliminating the non-correspondence problem) differs substantially.

\paragraph{Public Belief States in Two-Team Zero-Sum Games} As articulated in the introduction and Corollary \ref{cor:2t02}, our work is related to a recent body of literature \citep{carminati1,Zhang2022TeamBD,carminati2} showing that solving 2t0s games can be reduced to solving a 2p0s game by using intra-team policy announcements.
There has also been recent work leveraging this reduction to perform decision-time planning \citep{zhang2022subgame}.

\paragraph{Stackelberg Games} Public belief games with two time steps are closely related to Stackelberg games \citep{stackelberg,Schelling1960,Gibbons1992}.
A Stackelberg game is one in which a distinguished leader publicly commits to a strategy and a follower best responds to it, resulting in a bilevel optimization problem.
As with a public belief game, when a Stackelberg game is two player zero sum, Stackelberg equilibrium coincides with Nash equilibrium \textit{for the leader}, but the follower’s best response is generally highly exploitable in the game without public commitments.
While there exist tie-breaking procedures in Stackelberg literature (e.g., strong or weak Stackelberg equilibrium), they \textit{do not} resolve the non-correspondence issue.

\section{Conclusion and Future Work}

In this work, we provided a reduction from computing regularized equilibria of 2p0s games to computing regularized equilibria of PuB-AMGs.
We see this contribution as resolving an important gap in literature between common-payoff games and 2p0s games.

We see numerous impactful directions for future work.
The first involves comparing a high performance implementation of a regularized-objective-in-the-PuB-AMG approach to expert iteration \citep{exit,exit-thesis} to those of existing approaches \citep{brown2020rebel,pog};
while we have shown here that a regularized-objective-in-the-PuB-AMG approach possesses favorable properties in comparison to ReBeL \citep{brown2020rebel} and Player of Games \citep{pog}, verifying that these advantages manifest in practice would be a valuable contribution.
The second involves benchmarking a high performance implementation of a BAD-like \citep{bad} approach to learning in the regularized PuB-AMG; because our results open the door for the first time to such an approach, it is unknown how the performance of such an approach would compare against that of non-PBS-based model-free algorithms.
Third, by providing a simpler approach to working with PBSs in 2p0s games, our work provides further motivation for developing new approaches for approximating PBSs at scale; while \citet{bft} recently made progress in this direction by showing that fine-tuning can effectively approximate PBSs, amortized approximation of PBSs remains an open problem.
Finally, it may be possible to extend some weaker form of the results from our work to general-sum settings.

\section{Acknowledgements}

We thank Vickram Rajendran, Julien Perolat, Yiding Jiang, and Alexander Robey for helpful discussions.

\bibliography{example_paper}
\bibliographystyle{icml2023}

\newpage

\appendix

\onecolumn

\section{Definitions} \label{sec:def}

First, we formalize our definition of the correspondence mapping $\Pi^{\downarrow}$ in Algorithm \ref{alg:cor} below.

\begin{algorithm}[H]
    \caption{Correspondence Mapping $\Pi^{\downarrow}$}
    \label{alg:cor}
    \begin{algorithmic}
        \STATE {\bfseries Input:}{}{ $\tilde{\pi}$}
        \STATE $\text{queue} \gets [\tilde{s}^0]$
        \STATE $\pi \gets \{\}$            \WHILE{$\text{len}(\text{queue}) > 0$}
        \STATE $\tilde{s} \gets \text{queue.pop}()$
        \STATE $\tilde{a} \gets \tilde{\pi}(\tilde{s})$ \# Assume $\tilde{\pi}$ is deterministic.\footnotemark
        \FOR{$h \in \text{supp}(\tilde{s})$}
        \STATE $\pi(h_{\iota}) = \tilde{a}(h_{\iota})$
        \ENDFOR
        \FOR{$\tilde{s}' \in \text{supp}(\tilde{\mathcal{T}}(\tilde{s}, \tilde{a}))$} 
        \STATE $\text{queue.append}(\tilde{s'})$ 
        \ENDFOR
        \ENDWHILE
        \FOR{untouched $h_{\iota}$}
        \STATE Set $\pi(h_{\iota})$ arbitrarily. \# AOH is unreachable.
        \ENDFOR
        \STATE \textbf{return} $\pi$        
    \end{algorithmic}
\end{algorithm}
\footnotetext{This assumption is not required, but makes for cleaner presentation.
}

Next, we define a canonical choice function $\Pi^{\uparrow}$, which maps each joint policy to a corresponding PuB-AMG policy.

\begin{algorithm} [H]
    \caption{Canonical Choice Mapping $\Pi^{\uparrow}$} \label{alg:can}
    \begin{algorithmic}[1] 
        \STATE {\bfseries Input:}{}{$\pi$}
        \STATE $\tilde{\pi} \gets \{\}$
        \FOR{all $\tilde{s}$}
            \STATE $\tilde{a} \gets \{\}$
            \FOR{$h \in \text{supp}(\tilde{s})$}
                \STATE $\tilde{a}( h_{\iota}) = \pi(h_{\iota})$ \# Ignore the belief and do what $\pi$ does at $h_{\iota}$.
            \ENDFOR
            \STATE $\tilde{\pi}(\tilde{s}) = \tilde{a}$
        \ENDFOR
        \STATE \textbf{return} $\tilde{\pi}$
    \end{algorithmic}
\end{algorithm}

In short, Algorithm \ref{alg:can} yields a PuB-AMG policy in which the agents play according to $\pi$ irrespective of the public belief.
Therefore, we have that $\mathfrak{R}(\pi)= \mathfrak{\tilde{R}}(\Pi^{\uparrow}(\pi))$.
Note that, in contrast to the correspondence mapping, the canonical choice mapping is invariant to opponent policy.
Thus, we also allow $\Pi^{\uparrow}$ to be applied directly to individual player policies.

We also introduce some additional definitions.

\begin{definition}
For a minimax objective $\mathfrak{J}$, the value of the game under $\mathfrak{J}$ is $\max_{\pi_0'} \min_{\pi_1'} \mathfrak{J}(\pi_0', \pi_1') =  \min_{\pi_1'}\max_{\pi_0'} \mathfrak{J}(\pi_0', \pi_1')$.
\end{definition}

\begin{remark} \label{lem:game_val}
UG objectives guarantee a well-defined value.
This follows immediately from the fact that both players can guarantee the unique equilibrium value.
\end{remark}

\begin{definition}
For a minimax objective $\mathfrak{J}$, the best response value to $\pi_0$ under $\mathfrak{J}$ is $\min_{\pi_1'} \mathfrak{J}(\pi_0, \pi_1')$; analogously, the best response value to $\pi_1$ under $\mathfrak{J}$ is $\max_{\pi_0'} \mathfrak{J}(\pi_0', \pi_1)$.
We denote the best response to $\pi_i$ as $\text{BR}(\pi_i)$.
A policy is part of a Nash equilibrium if it achieves the value of the game against a best response.
\end{definition}

\begin{definition}
For a minimax objective $\mathfrak{J}$, the exploitability $\pi$ under $\mathfrak{J}$ is defined as: \[\text{expl}(\pi) = \frac{-\min_{\pi_1'} \mathfrak{J}(\pi_0, \pi_1') + \max_{\pi_0'} \mathfrak{J}(\pi_0', \pi_1)}{2}.\]
A joint policy is a Nash equilibrium if it has exploitability zero. 
\end{definition}

\begin{definition}
For a minimax objective $\mathfrak{J}$ induced by 
\[\mathfrak{R} \colon (h, a, \delta) \mapsto 
\begin{cases}
\mathcal{R}(h, a) - \psi(\delta, h_{\iota}) & \iota = 0\\
\mathcal{R}(h, a) + \psi(\delta, h_{\iota}) & \iota = 1,
\end{cases}\] the action value for action $a$ at AOH $h_{\iota}^{t}$ under joint policy $\pi$ is
\[Q(h_{\iota}, a) = (-1)^{\mathbb{I}[\iota=1]}\mathbb{E}_{\pi} \left[\mathfrak{R}(H^t, A^t, a \mapsto \mathbb{I}[a^t=a]) + \sum_{t' > t}^T \mathfrak{R}(H^{t'}, A^{t'}, \pi(A^{t'})) \mid h_{\iota}^{t}, a^t \right].\]
In words, it is the expected future value to the acting player for taking $a$ at $h_{\iota}^{t'}$ assuming that both players have played according to $\pi$ up until now and will continue to play according to $\pi$ hereinafter.
 \end{definition}
\section{Theory}

We now detail the proofs of our theoretical results.

\subsection{Non-Correspondence of Nash Equilibria} \label{sec:worst_case}

\begin{figure}[H]
    \centering
    \includegraphics[width=\textwidth]{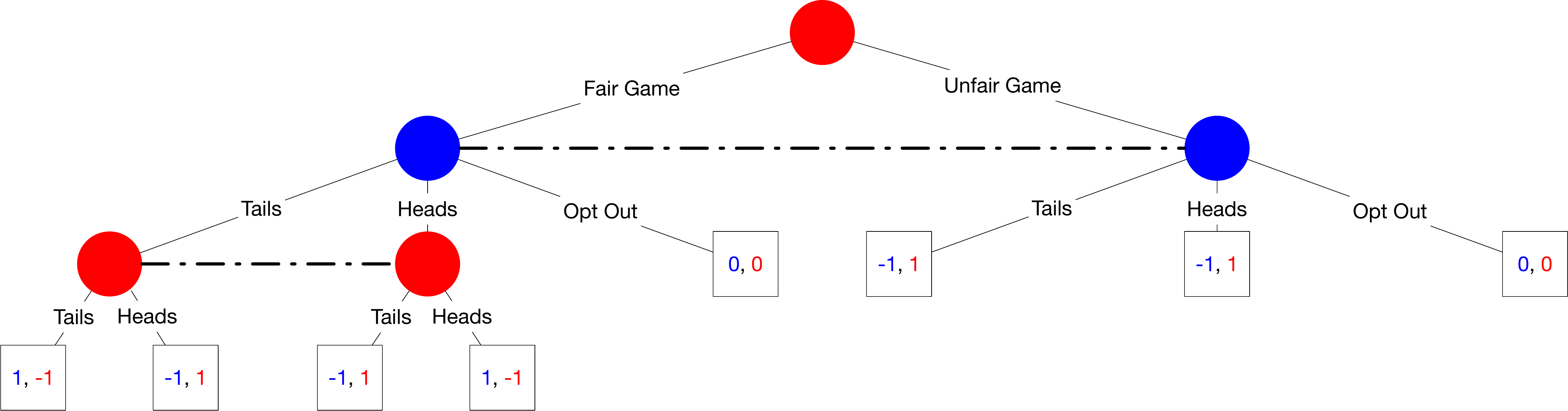}
    \caption{A rigged and adversarial variant of matching pennies.}
    \label{fig:worst_case}
\end{figure}

\textbf{Proposition \ref{prop:worst_case}}
A PuB-AMG Nash equilibrium $\tilde{\pi}$ may correspond with a joint policy $\Pi^{\downarrow}(\tilde{\pi})$ that is maximally exploitable.

\begin{proof}
We show that this worst case can be realized in the 2p0s rigged adversarial matching pennies game depicted in Figure \ref{fig:worst_case}.
The game starts with two options for the red player: it can either decide to make the game fair or to rig the game.
Then, without having observed the red player's decision, the blue player decides whether to opt out of the game altogether, in which case both players receive a payout of $0$, or to play adversarial matching pennies.
If the blue player opts in and the game is rigged, the red player receives a payout of $1$ independent of the blue player's selection.
If the blue player opts in and the game is not rigged, the blue player receives a payout of $1$ if the players select the same side of the coin; otherwise, if the players selected opposite sides of the coin, the red player receives a payout of $1$.

In the game, the blue player's only Nash equilibrium strategy is to opt out with probability one.
The red player's Nash equilibria strategies require at least one of i) rigging the game with probability one and ii) mixing 50-50 between tails and heads.

Now, consider the following PuB-AMG policy, where superscript denotes time step within the game:
\begin{itemize}[leftmargin=*]
    \item $\tilde{\pi}^0(\emptyset) = (\text{Fair}, \text{Unfair}) \mapsto (1, 0)$.
    \item $\tilde{\pi}^1(\tilde{\pi}^0) = 
    \begin{cases}
    (\text{Tails}, \text{Heads}, \text{Out}) \mapsto (1/2, 1/2, 0) & \tilde{\pi}^0(\text{Fair}) = 1\\
    (\text{Tails}, \text{Heads}, \text{Out}) \mapsto (0, 0, 1) & \text{otherwise.}
    \end{cases}$
    \item 
    $\tilde{\pi}^2(\tilde{\pi}^0, \tilde{\pi}^1) =
    \begin{cases}
    (\text{Tails}, \text{Heads}) \mapsto (0, 1) & \tilde{\pi}^1(\text{Tails}) \geq  1/2\\
    (\text{Tails}, \text{Heads}) \mapsto (1, 0) & \text{otherwise.}
    \end{cases}$
\end{itemize}
We claim that $\tilde{\pi}$ is a PuB-AMG Nash equilibrium.
To see this, first consider that the expected return is $0$: the red player always opts in, the blue player mixes evenly between heads and tails, and the red player always selects tails.
Next consider that the red player has no incentive to deviate to an unfair game, because the blue player will opt out, yielding an expected return of zero.
Also consider the blue player has no incentive to place additional mass on opting out, as it yields an expected return of zero.
Furthermore, the blue player has no incentive to select a different mixture of heads and tails, as doing so will decrease its expected return since the red player best responds at the final time step.
Lastly, consider that the red player is best responding at the final time step and, therefore, has no incentive to deviate.

Then, consider that the corresponding policy $\pi=\Pi^{\downarrow}(\tilde{\pi})$ is as follows:
\begin{itemize}
    \item $\pi^0 \colon (\text{Fair}, \text{Unfair}) \mapsto (1, 0)$.
    \item $\pi^1 \colon (\text{Tails}, \text{Heads}, \text{Out}) \mapsto (1/2, 1/2, 0)$.
    \item $\pi^2 \colon (\text{Tails}, \text{Heads}) \mapsto (0, 1)$.
\end{itemize}
We claim that this policy is maximally exploitable.
To see this, consider that a red player that always rigs the game achieves an expected return of one against the blue player's policy, and consider that a blue player that always selects heads achieves an expected return of one against the red player's policy. 
\end{proof}

\subsection{ Correspondence of Uniqueness-Guaranteeing Equilibria} \label{sec:cor}

To prove Theorem \ref{thm:correspondence}, we first require some lemmas.

\begin{lemma} \label{lem:value_guar}
The best response value to $\pi_i$ in the original game is equal to the best response value of $\Pi^{\uparrow}(\pi_i)$ in PuB-AMG.
\end{lemma}
 
\begin{proof}
This follows because player $-i$ has no mechanism to exploit $\Pi^{\uparrow}(\pi_i)$ beyond that of the original game, since $\Pi^{\uparrow}(\pi_i)$ ignores belief information.

More formally, consider
\begin{align*}
\min_{\tilde{\pi}_1'} \mathfrak{\tilde{J}}(\Pi^{\uparrow}(\pi_0), \tilde{\pi}_1')
&= \mathfrak{\tilde{J}}(\Pi^{\uparrow}(\pi_0), \text{BR}(\Pi^{\uparrow}(\pi_0)))\\
&= \mathfrak{J}(\Pi^{\downarrow}(\Pi^{\uparrow}(\pi_0), \text{BR}(\Pi^{\uparrow}(\pi_0))))\\
&= \mathfrak{J}(\Pi^{\downarrow}(\Pi^{\uparrow}(\pi_0), \text{BR}(\Pi^{\uparrow}(\pi_0))))_0, \Pi^{\downarrow}(\Pi^{\uparrow}(\pi_0), \text{BR}(\Pi^{\uparrow}(\pi_0))))_1)\\
&= \mathfrak{J}(\pi_0, \Pi^{\downarrow}(\Pi^{\uparrow}(\pi_0), \text{BR}(\Pi^{\uparrow}(\pi_0))))_1)\\
&\geq \mathfrak{J}(\pi_0, \text{BR}(\pi_0))\\
&= \min_{\pi_1'}\mathfrak{J}(\pi_0, \pi_1').
\end{align*}
The first equality follows by definition of the best response function $\text{BR}$.
The second equality because $\Pi^{\downarrow}$ preserves expected return.
The third equality is notational expansion.
The fourth equality follows because $\pi_0$ and $\Pi^{\downarrow}(\Pi^{\uparrow}(\pi_0), \text{BR}(\Pi^{\uparrow}(\pi_0))))_0$ can only differ at AOHs that are not reached when playing against $\Pi^{\downarrow}(\Pi^{\uparrow}(\pi_0), \text{BR}(\Pi^{\uparrow}(\pi_0))))_1$.
The inequality and final equality follow by definition of best response.

Also, consider 
\begin{align*}
\min_{\pi_1'}\mathfrak{J}(\pi_0, \pi_1')
&= \mathfrak{J}(\pi_0, \text{BR}(\pi_0))\\
&= \mathfrak{\tilde{J}}(\Pi^{\uparrow}(\pi_0), \Pi^{\uparrow}(\text{BR}(\pi_0)))\\
&\geq \mathfrak{\tilde{J}}(\Pi^{\uparrow}(\pi_0), \text{BR}(\Pi^{\uparrow}(\pi_0)))\\
&= \min_{\tilde{\pi}_1'}\mathfrak{\tilde{J}}(\Pi^{\uparrow}(\pi_0), \tilde{\pi}_1).
\end{align*}
The first equality follows by definition of the best response function $\text{BR}$.
The second equality follows because $\Pi^{\uparrow}$ preserves expected return.
The inequality and final equality follows by definition of best response.

These two inequalities can only be true if $\min_{\pi_1'} \mathfrak{\tilde{J}}(\Pi^{\uparrow}(\pi_0), \pi_1')=\min_{\pi_1'}\mathfrak{J}(\pi_0, \pi_1')$.
An analogous argument shows the same result for $\pi_1$.
\end{proof}

\begin{corollary} \label{cor:expl}
The exploitability of $\pi$ in the original game is equal to the exploitability of $\Pi^{\uparrow}(\pi)$ in PuB-AMG.
\end{corollary}
\begin{proof}
This follows immediately from Lemma \ref{lem:value_guar} and the fact that exploitability is defined in terms of best response values.
\end{proof}

\begin{corollary}
 \label{cor:value}
 The value of the PuB-AMG is well defined and equal to that of the original game.
\end{corollary}
\begin{proof}
Note that it suffices to show that PuB-AMGs are guaranteed to have an equilibrium with the same expected return as the equilibrium of the original game.
Then consider $\tilde{\pi}=\Pi^{\uparrow}(\pi)$, where $\pi$ is an equilibrium.
Then, since $\pi$ is an equilibrium and, per Corollary \ref{cor:expl}, $\Pi^{\uparrow}$ preserves exploitability, $\tilde{\pi}$ is an equilibrium.
Additionally, since $\Pi^{\uparrow}$ preserves expected return, the original game and the PuB-AMG possess equilibria $\pi$ and $\tilde{\pi}$, respectively, that yield the same expected return.
\end{proof}

We are now ready to prove our two main lemmas.
\begin{lemma} \label{lem:ext}
Let $\pi$ be the equilibrium of a UG objective.
Let $\tilde{s}$ define a subgame of the original game induced by playing $\pi$ for some number of steps.
Then the unique equilibrium $\pi^{\tilde{s}}$ of the subgame, considered as an independent game, is the restriction $\pi^{\mid \tilde{s}}$ of $\pi$ to the subgame.
\end{lemma}

\begin{proof}
If $\pi^{\mid \tilde{s}}$ is not an equilibrium of the subgame, then it must be exploitable in the subgame.
This means that either
\[\min_{\pi_1'}\mathfrak{J}^{\tilde{s}}(\pi^{\mid \tilde{s}}_0, \pi_1') < \mathfrak{J}^{\tilde{s}}(\pi^{\mid \tilde{s}}_0, \pi^{\mid \tilde{s}}_1) \quad \text{ or } \quad \mathfrak{J}^{\tilde{s}}(\pi^{\mid \tilde{s}}_0, \pi^{\mid \tilde{s}}_1) < \max_{\pi_0'}\mathfrak{J}^{\tilde{s}}(\pi_0', \pi^{\mid \tilde{s}}_1)\]
Without loss of generality, assume the former.
Let $\pi^{\text{br}}_1 = \text{argmin}_{\pi_1'} \mathfrak{J}^{\tilde{s}}(\pi^{\mid \tilde{s}}, \pi_1')$; 
let $\mathcal{P}^{\pi}(\tilde{s})$ represent the probability of reaching $\tilde{s}$ using policy $\pi$;
let $t$ be the time step corresponding to $\tilde{s}$ and let $\mathcal{J}^{<t}$ denote the expected return prior to time $t$.
Further, let $\tilde{s}' \neq \tilde{s}$ range over the possible subgames entered at time $t$ if $\tilde{s}$ is not entered.
Then
\begin{align*}
\mathfrak{J}(\pi_0, \pi_1)
&= \mathfrak{J}^{<t}(\pi_0, \pi_1) + \mathcal{P}^{\pi}(\tilde{s}) \mathfrak{J}^{\tilde{s}}(\pi^{\mid \tilde{s}}_0, \pi_1^{\mid \tilde{s}}) + \sum_{\tilde{s}' \neq \tilde{s}}\mathcal{P}^{\pi}(\tilde{s}') \mathfrak{J}^{\tilde{s}'}(\pi^{\mid \tilde{s}'}_0, \pi_1^{\mid \tilde{s}'})\\
&> \mathfrak{J}^{<t}(\pi_0, \pi_1) + \mathcal{P}^{\pi}(\tilde{s}) \mathfrak{J}^{\tilde{s}}(\pi^{\mid \tilde{s}}_0, \pi^{\text{br}}_1) + \sum_{\tilde{s}' \neq \tilde{s}}\mathcal{P}^{\pi}(\tilde{s}') \mathfrak{J}^{\tilde{s}'}(\pi^{\mid \tilde{s}'}_0, \pi_1^{\mid \tilde{s}'})\\
&= \mathfrak{J}(\pi_0, [\pi^{\text{br}}_1, \pi_1^{\mid -\tilde{s}}]).
\end{align*}
Here, the first line decomposes the expected return into 1) that which is accrued prior to time $t$, 2) that which is accrued in subgame $\tilde{s}$, and 3) that which is accrued after time $t$ outside of subgame $\tilde{s}$.
The second line invokes our assumption that $\pi^{\mid \tilde{s}}_1$ does not achieve the best response value against $\pi^{\mid \tilde{s}}_0$ and $\mathcal{P}^{\pi}(\tilde{s}) > 0$.
The third line re-assembles the expected return, where we use $[\pi^{\text{br}}_1, \pi_1^{\mid -\tilde{s}}]$ to denote a policy that plays $\pi_1$ outside $\tilde{s}$ and $\pi^{\text{br}}_1$ inside $\tilde{s}$.

In total, we have shown that if $\pi^{\mid \tilde{s}}$ is not an equilibrium in the subgame induced by $\tilde{s}$, then $\pi$ is not an equilibrium because $\pi_1$ is not a best response.
Thus, $\pi^{\mid \tilde{s}}$ must be an equilibrium of the subgame.
Therefore, because $\mathfrak{J}$ is UG, we must have $\pi^{\tilde{s}}=\pi^{\mid \tilde{s}}$.
\end{proof}

\begin{lemma} \label{lem:first}
If $\tilde{\pi}$ is an equilibrium of the PuB-AMG, then the decision rule for the first time step $\Pi^{\downarrow}(\tilde{\pi})^0$ must be part of an equilibrium policy in the original game. 
\end{lemma}
\begin{proof}
Without loss of generality, assume that $\iota=0$ at the first time step.
Also, use $\pi_0^{-0}$ to denote the part of $\pi_0$ that is relevant after the first time step.
Also, let $[\pi^{\prime0}, \pi_0^{\prime-0}]$ denote a policy for $i=0$ that plays according to $\pi^{\prime0}$ at the first time step and $\pi_0^{\prime-0 }$ otherwise.
Let $\tilde{\pi}$ be an equilibrium of the PuB-AMG.
Then observe
\begin{align}
\max_{\pi_0'} \min_{\pi_1'}  \mathfrak{J}(\pi_0', \pi_1') 
&= \max_{\tilde{\pi}_0'} \min_{\tilde{\pi}_1'}  \mathfrak{\tilde{J}}(\tilde{\pi}_0', \tilde{\pi}_1')\\
&= \max_{\tilde{\pi}_0^{\prime-0}} \min_{\tilde{\pi}_1'}  \mathfrak{\tilde{J}}([\tilde{\pi}^0, \tilde{\pi}_0^{\prime-0 }], \tilde{\pi}_1)\\
&= \max_{\tilde{\pi}_0^{\prime-0}} \min_{\tilde{\pi}_1'}  \mathfrak{J}(\Pi^{\downarrow}([\tilde{\pi}^0, \tilde{\pi}_0^{\prime-0 }], \tilde{\pi}_1'))\\
&= \max_{\tilde{\pi}_0^{ \prime-0}} \min_{\tilde{\pi}_1'}  \mathfrak{J}([\Pi^{\downarrow}(\tilde{\pi})^0, \Pi^{\downarrow}([\tilde{\pi}^0, \tilde{\pi}_0^{\prime-0 }], \tilde{\pi}_1')_0^{-0}], \Pi^{\downarrow}([\tilde{\pi}^0, \tilde{\pi}_0^{\prime-0 }], \tilde{\pi}_1')_1)\\
&= \max_{\pi_0^{\prime-0 }} \min_{\pi_1'}  \mathfrak{J}([\Pi^{\downarrow}(\tilde{\pi})^0, \pi_0^{\prime-0 }], \pi_1')\\
&= \min_{\pi_1'} \mathfrak{J}([\Pi^{\downarrow}(\tilde{\pi})^0, \text{arg max}_{\pi_0^{\prime-0 }} \min_{\pi_1''}  \mathfrak{J}([\Pi^{\downarrow}(\tilde{\pi})^0, \pi_0^{\prime-0}], \pi_1'')], \pi_1').
\end{align}
Here, the first equality follows by Corollary \ref{cor:value};
the second equality follows because $\tilde{\pi}^0$ is part of an equilibrium; 
the third equality follows because $\mathfrak{\tilde{J}}(\tilde{\pi}') = \mathfrak{J}(\Pi^{\downarrow}(\tilde{\pi}'))$;
the fourth line equality follows because the image of the correspondence mapping for the first time step is invariant to the PuB-AMG policy at later time steps;
the fifth line follows because each player can express any policy in the original game through $\Pi^{\downarrow}$, up to reachability, and because changes over unreachable AOHs do not change the expected return;
the sixth line follows because the evaluation of an argmax is equal to the max.

This chain of equalities shows that the best response value to \[[\Pi^{\downarrow}(\tilde{\pi})^0, \text{arg max}_{\pi_0^{\prime-0 }} \min_{\pi_1''}  \mathfrak{J}([\Pi^{\downarrow}(\tilde{\pi})^0, \pi_0^{\prime-0}], \pi_1'')]\]
is equal to the value of the game.
Thus, $\Pi^{\downarrow}(\tilde{\pi})^0$ is part of an equilibrium.
\end{proof}

\textbf{Theorem \ref{thm:correspondence}} If $\tilde{\pi}$ is an equilibrium of the PuB-AMG induced by a UG objective, then its corresponding policy $\Pi^{\downarrow}(\tilde{\pi})$ is an equilibrium in the original game.

\begin{proof}
Lemma \ref{lem:first} shows this to be true for the first time step.
Now assume this is true up to time step $t$ and consider time step $t+1$.
Then, for a particular reachable $\tilde{s}^{t+1}$, the PuB-AMG subgame starting at this point is the PuB-AMG of the subgame of the original game starting from $\tilde{s}^{t+1}$.
Thus, the PuB-AMG strategy for $\tilde{s}^{t+1}$ must correspond to an equilibrium of the subgame of the original game, as per Lemma \ref{lem:first}.
Furthermore, because the minimax objective is UG, the equilibrium strategy of the subgame of the original game must the unique restriction of the equilibrium of the original game to that subgame, as per Lemma \ref{lem:ext}.
\end{proof}

\subsection{Continuity in the PuB-AMG with Uniqueness-Guaranteeing Objectives} \label{sec:cont}

\begin{lemma} \label{lem:mm}
Let $f_1$, $f_2$ be real-valued continuous functions with shared compact domain $\mathcal{X}\times \mathcal{Y}$.
Furthermore, assume their max-min values are attained and the following inequality holds for any $(x,y) \in \mathcal{X}\times \mathcal{Y}$:
\[|f_1(x, y)- f_2(x, y)| < \epsilon.\] Then it follows that 
\[|[\max_{x \in \mathcal{X}} \min_{y \in \mathcal{Y}} f_1(x, y)]-[\max_{x \in \mathcal{X}} \min_{y \in \mathcal{Y}} f_2(x, y)]| <
\epsilon.\]
\end{lemma}

\begin{proof}
Note that by assumption we have for all $(x,y)$ within the domains of $f_1, f_2$ it holds
\begin{align}
    &f_1(x,y) \leq f_2(x,y) + \epsilon,\label{eq:eps_upper}\\ 
    &f_2(x,y) \leq f_1(x,y) + \epsilon.\label{eq:eps_upper_2}
\end{align}

Therefore,
\begin{align*}
    \max_{x \in \mathcal{X}} \min_{y \in \mathcal{Y}} f_1(x,y) &= \min_{y \in \mathcal{Y}} f_1(x_\ast,y)  \mbox{ for }  x_\ast \in \arg \max_{x \in \mathcal{X}} \min_{y\in\mathcal{Y}} f_1(x,y)\\
     &\leq  \min_{y \in \mathcal{Y}} f_2(x_\ast,y) + \epsilon\\
     &\leq \max_{x \in \mathcal{X}} \min_{y \in \mathcal{Y}} f_2(x,y) +\epsilon.
\end{align*}
Where the first inequality is due to taking the min of both sides of \eqref{eq:eps_upper}.
Following the same steps starting with $f_2$ and using \eqref{eq:eps_upper_2} gives
\[
\max_{x \in \mathcal{X}} \min_{y \in \mathcal{Y}} f_2(x,y) \leq  \max_{x \in \mathcal{X}} \min_{y \in \mathcal{Y}} f_1(x,y) +\epsilon.
\]
These two inequalities together yield the result.
\end{proof}

\textbf{Theorem \ref{thm:val_cont}.}
\textit{Let $\mathfrak{J}$ guarantee the existence of an equilibrium in all subgames. Then the PuB-AMG subgame perfect equilibrium value function is a continuous function from the space of PBSs to real values.}

\begin{proof}
Fix $\epsilon > 0$. Let $b$ and $b'$ differ in total variation distance by less than 
\[\delta = \frac{\epsilon}{2 \mathfrak{M}}.\]
Then observe that, for a joint policy $\pi$, we have that
\begin{align*}
|v_{\pi}(b) - v_{\pi}(b')|
&= |\sum_{h} b(h) v_{\pi}(h) - b'(h) v_{\pi}(h)|\\
&\leq \sum_{h} |b(h) v_{\pi}(h) - b'(h) v_{\pi}(h)|\\
&= \mathfrak{M} \sum_{h}  |b(h) - b'(h)|\\
&< 2 \mathfrak{M} \delta\\
&= \epsilon,
\end{align*}
where $v_{\pi}(b)$ is the expected return under $\mathfrak{J}$ to playing $\pi$ starting from the subgame defined by $b$.

Then
\begin{align*}
|\tilde{v}_{\ast}(b) - \tilde{v}_{\ast}(b')| &= |[\max_{\tilde{\pi}_0} \min_{\tilde{\pi}_1} \tilde{v}_{\tilde{\pi}}(b)] - [\max_{\tilde{\pi}_0} \min_{\tilde{\pi}_1} \tilde{v}_{\tilde{\pi}}(b')]|\\
&= 
| [\max_{\pi_0} \min_{\pi_1} v_{\pi}(b)] - [\max_{\pi_0} \min_{\pi_1} v_{\pi}(b')]|\\
&< \epsilon.
\end{align*}
The first equality follows by definition of $\tilde{v}_{\ast}$.
The second equality follows from Corollary \ref{cor:value}.
The third equality follows from Lemma \ref{lem:mm}.
\end{proof}

\begin{remark}
Theorem \ref{thm:val_cont} shows that continuous objectives yield continuous value functions in the PuB-AMG, even if the objective is not UG.
\end{remark}

Next, we prove the continuity of the equilibrium policy mapping under UG objectives.

\textbf{Theorem \ref{thm:pol_cont}.} \textit{Let $\mathfrak{J}$ be a UG objective be induced by a continuous $\mathfrak{R}$.
Then the PuB-AMG subgame perfect equilibrium induced by $\mathfrak{J}$ is a continuous function from the space of PBSs to the space of public decision rules.}

\begin{proof}
Consider that  $q_{\ast} \colon \tilde{s}, \tilde{a} \mapsto \mathfrak{\tilde{R}}(\tilde{s}, \tilde{a}) + \mathbb{E}_{\tilde{S}' \sim \tilde{\mathcal{T}}(\tilde{s}, \tilde{a})} \tilde{v}_{\ast}(\tilde{S}')$ is continuous because $\mathfrak{\tilde{R}}$ is continuous (since $\mathfrak{R}$ is continuous), $\tilde{\mathcal{T}}$ is continuous by construction, and $\tilde{v}_{\ast}$ is continuous by Theorem \ref{thm:val_cont}.
Then the maximum theorem states that $\pi_{\ast} \colon \tilde{s} \mapsto \text{arg max}_{\tilde{a}'} q_{\ast}(\tilde{s}, \tilde{a}')$ is an upper hemicontinuous function.
Finally, because $\mathfrak{J}$ is UG, we have that $\pi_{\ast}$ is single valued.
The result follows from the fact that single-valued upper hemicontinuous functions are continuous.
\end{proof}

In contrast, for non-UG objectives, equilibrium policy mapping is not necessarily continuous.

\begin{proposition}
Let $\mathfrak{J}$ be an objective be induced by the unregularized reward $\mathcal{R}$. Then the PuB-AMG subgame perfect equilibrium is not generally continuous.
\end{proposition}

\begin{proof}
Consider the adversarial variant of matching pennies described in Figure \ref{fig:adv_mp_app}.
Let $p$ denote the probability with which the blue player selected heads.
Let $q$ denote the probability with which the red player selects heads.
Then the red player's equilibrium policy is:
\begin{itemize}
    \item If $p > 1/2$, $q = 0$
    \item If $p = 1/2$, $q \in [0, 1]$.
    \item If $p < 1/2$, $q = 1$.
\end{itemize}
The result follows from the fact that the mapping from $p$ to $q$ is not continuous.
\end{proof}

\begin{figure*}
    \centering
    \includegraphics[width=\linewidth]{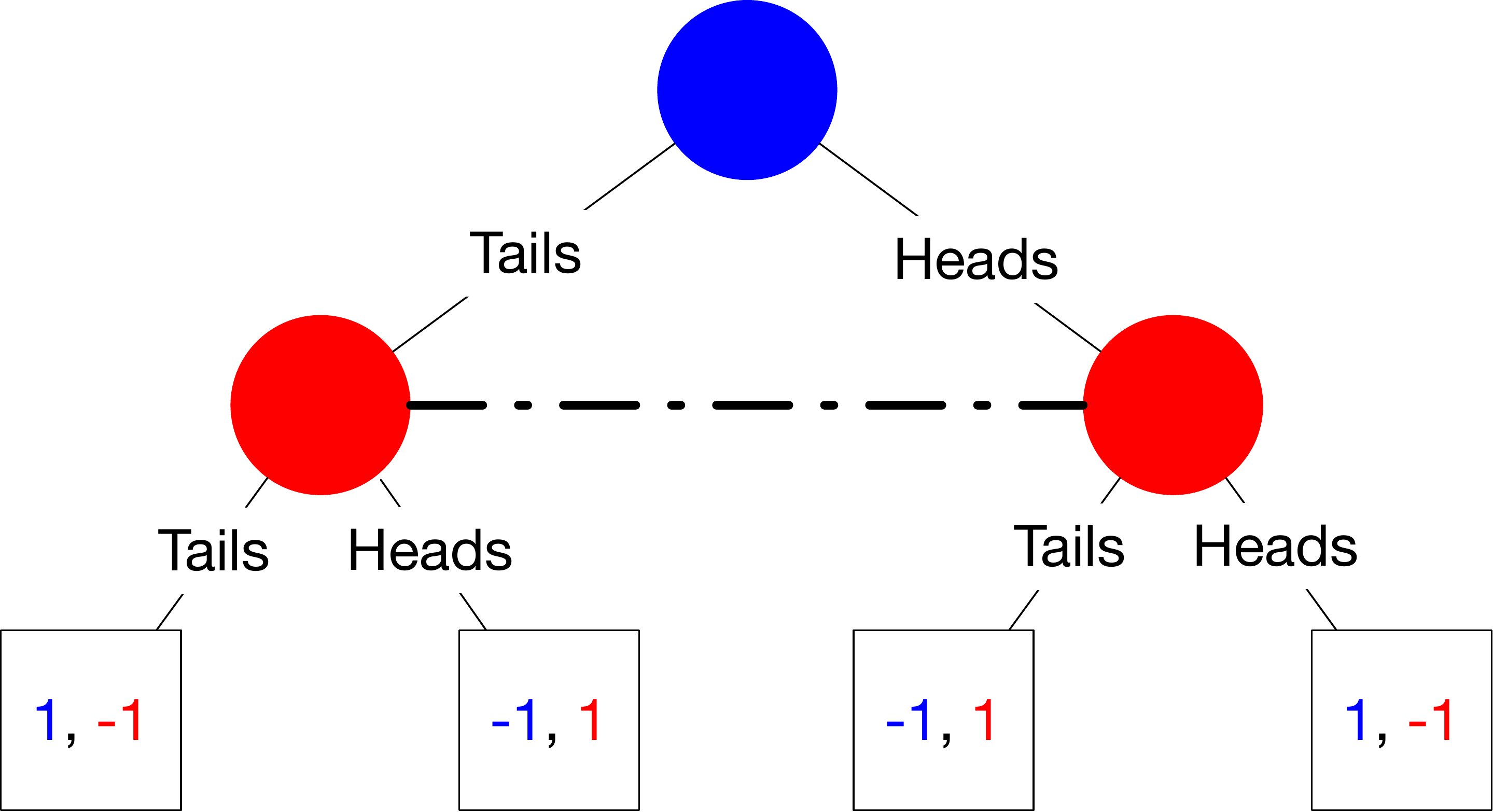}
    \caption{An adversarial variant of the game matching pennies.}
    \label{fig:adv_mp_app}
\end{figure*}

\subsection{Sufficient Conditions for Uniqueness Guaranteeing} \label{sec:suff}

\textbf{Proposition \ref{prop:close}.} Let $\mathfrak{J}$ be a MiniMaxKL objective parameterized by a reference policy $\rho$, placing at least $\epsilon$ probability on every action, and regularization parameter $\alpha$.
Then the exploitability of the MiniMaxKL equilibrium is bounded by $\alpha T |\log \epsilon|$, where $T$ is the horizon of the game.

\begin{proof}
Consider that, at each time step, the component of a player's reward arising from the KL term is at most
\begin{align*}
\max_{\delta \in \Delta(\mathbb{A})} \alpha \text{KL}(\delta, \rho(h_{\iota}))
&= \max_{\delta \in \Delta(\mathbb{A})} \alpha(\mathcal{H}(\delta, \rho(h_{\iota})) - \mathcal{H}(\delta))\\
&\leq \max_{\delta \in \Delta(\mathbb{A})} \alpha \mathcal{H}(\delta, \rho(h_{\iota}))\\
&= \max_{\delta \in \Delta(\mathbb{A})} - \alpha \sum_a \delta(a) \log \rho(h_{\iota}, a)\\
&= \max_{\delta \in \Delta(\mathbb{A})} |\alpha  \sum_a \delta(a) \log \rho(h_{\iota}, a)|\\
&= \max_a |\alpha  \log \rho(h_{\iota}, a)|\\
&\leq \alpha |\log \epsilon |.
\end{align*}
Here, the first line follows from the fact that KL divergence can be decomposed into a sum of cross-entropy and entropy; the second line follows because entropy is positive; the third line is definitional; the fourth line follows because taking the absolute value of a negative number is equivalent to multiplying by negative one; the fifth line follows because weighted sums are maximized by placing all the weight on the largest value; the sixth line follows by assumption.

Because the length of the game is bounded by $T$, the expected return under the regularized objective can differ from the expected return by no more than $T \alpha |\log \epsilon |$.
Now, let $\pi^{\ast}$ be the equilibrium under the regularized objective and let $\pi'$ be a best response under the unregularized objective.
Then we have
\begin{align*}
\text{expl}(\pi^{\ast})
&= \frac{-\mathcal{J}(\pi^{\ast}_0, \pi_1') + \mathcal{J}(\pi_0', \pi^{\ast}_1)}{2}\\ 
&\leq \frac{-\mathfrak{J}(\pi^{\ast}_0, \pi_1') + \mathfrak{J}(\pi_0', \pi^{\ast}_1)}{2} + \alpha T | \log \epsilon |\\
&\leq \alpha T | \log \epsilon |,
\end{align*}
where the second inequality follows because the regularized equilibrium is unexploitable under the regularized objective.
\end{proof}

\begin{theorem} \label{thm:suff}
Consider an objective $\mathfrak{J}$ induced by 
\[\mathfrak{R} \colon (h, a, \delta) \mapsto 
\begin{cases}
\mathcal{R}(h, a) - \psi(\delta, h_{\iota}) & \iota = 0\\
\mathcal{R}(h, a) + \psi(\delta, h_{\iota}) & \iota = 1
\end{cases}\]
and define a policy greedification function
\[g \colon [-\mathfrak{M}, \mathfrak{M}]^{|\mathbb{A}|} \times \mathbb{H}_{\iota} \to \Delta(\mathbb{A})\]
where $\mathfrak{M} \in \mathbb{R}$ is the maximum of the absolute values of the expected returns of $\mathfrak{J}$ and where
\[g \colon q, h_{\iota} \mapsto 
\text{arg max}_{\delta \in \Delta(\mathbb{A})} \langle \delta, q \rangle - \psi(\delta, h_{\iota}).\]
In words, for each AOH at which a player acts $h_{\iota}$, $g$ maps possible regularized action values $q$ to the policy that is greedy with respect to the regularized objective under those regularized action values.

If, for all $h_{\iota}$, $\psi(\cdot, h_{\iota})$ is continuous and $g(\cdot, h_{\iota})$ is i) well defined, ii) continuous, and iii) has an interior image, then the objective $\mathfrak{J}$ is UG.
\end{theorem}

\begin{proof}
First, we show that such an equilibrium is guaranteed to exist.
Let $\mathcal{F} \colon [-\mathfrak{M}, \mathfrak{M}]^{|\mathcal{H}_{\iota}||\mathcal{A}|} \to [-\mathfrak{M}, \mathfrak{M}]^{|\mathcal{H}_{\iota}||\mathcal{A}|}$ be a function that maps each vector $[q_{h_{\iota}}]_{h_{\iota}}$ to the action-value vector for the joint policy dictated by the application of $g$ to $[(q_{h_{\iota}}, h_{\iota})]_{h_{\iota}}$.
Note that $\mathcal{F}$ is well defined---i.e., the ensuant action values are always well defined---because $g$ maps to the interior, so every history is reached with positive probability.
Also note that this function is continuous, by the continuity of $g$ and $\psi$, and single valued because $g$ is single valued.
Thus, because $[-\mathfrak{M}, \mathfrak{M}]^{|\mathcal{H}_{\iota}||\mathcal{A}|}$ is compact and convex, by Brouwer's fixed point theorem, a fixed point must exist.
The policy corresponding to these fixed-point action values is an equilibrium.
This follows because, by backward induction, each player is optimally responding to the other, holding the other fixed.

Now we show that there is a unique equilibrium.
Note that, for any fixed opponent, the optimal policy at any decision point reached with positive probability must be full support because $g$'s image is within the interior.
By forward induction, this means that every equilibrium must be full support at every decision point.
Now, note that, by backward induction, the best responses to full support policies are unique because $g$ is single valued with an interior image.
In aggregate, these two things show that any equilibrium is strict---i.e., the only best response to one part of the equilibrium is the other part of the equilibrium.
Now, assume there exist two distinct equilibria $\pi$ and $\pi'$.
Without loss of generality, assume that $\pi_0$ performs at least as well as $\pi_0'$ against $\pi_1$.
If $\pi_0$ performs equally well, there is a contradiction because $\pi_0'$ is not the unique best response.
If $\pi_0$ outperforms $\pi'$, there is a contradiction because $\pi'$ is not at equilibrium.
Thus, the equilibrium must be unique.

The result follows because this proof also holds for every subgame of the original game.
\end{proof}

\begin{remark}
The premises of Theorem \ref{thm:suff} are satisfied if $\psi(\cdot, h_{\iota})$ is bounded and is strictly convex and differentiable on its interior with $\lim_{\delta \to \delta'}||\nabla_{\delta} \psi (\delta, h_{\iota})||=+\infty$ for $\delta'$ on the boundary of $\Delta(\mathbb{A})$.
\end{remark}

\begin{remark}
One example of an objective covered by Theorem \ref{thm:suff}, but not by Theorem \ref{thm:perolat}, is that which is induced by setting $\psi(\cdot, h_{\iota})$ to a sum of a KL divergence to an interior point and a bounded differentiable convex function.
\end{remark}

\begin{remark}
The equilibria of objectives satisfying the premises of Theorem \ref{thm:suff} can achieve arbitrarily low exploitabilities by similar reasoning as Proposition \ref{prop:close}.
\end{remark}

\section{Experiments} \label{sec:exp}

\subsection{Magnetic Mirror Descent}

In our experiments, we use magnetic mirror descent (MMD) \citep{mmd} as our game solver.
In the instance of MMD we use, updates are of the form
\begin{align}\label{rl-mmd}
\pi_{t+1} = \text{argmax}_{\pi} \mathbb{E}_{A \sim \pi} q_{\pi_t}(A) + \alpha \mathcal{H}(\pi) - \frac{1}{\eta}\text{KL}(\pi, \pi_t)
\end{align}
where $\pi_t$ is the current policy and $q_{\pi_t}$ is the MiniMaxEnt Q-value vector for time $t$.
This update possesses the closed form
\begin{align} \label{mmd-update}
\pi_{t+1} \propto [\pi_t e^{\eta q_{\pi_t}}]^{\frac{1}{1 + \alpha \eta}}.
\end{align}

The fixed point of equation (\ref{mmd-update}) is a policy satisfying
\begin{align} \label{eq:greedy}
\pi_{\ast} = \text{arg max}_{\pi} \mathbb{E}_{A \sim \pi} q_{\pi_{\ast}}(A) + \alpha \mathcal{H}(\pi) \propto e^{q_{\pi_{\ast}} / \alpha}.
\end{align}
\subsection{Tabular PuB-AMG Policies}

In PuB-AMGs, the state space is continuous.
Thus, it may not be possible to express a fully specified PuB-AMG policy in tabular form.
We describe how we handle this issue for perturbed rock-paper-scissors and Kuhn poker, respectively, in subsequent subsections.

In both settings, we solve the games using full feedback, meaning that we compute exact Q-values and update the policy for every AOH.

\subsubsection{Perturbed Rock-Paper-Scissors}

In perturbed rock-paper-scissors, the first moving player's state space is trivial; thus, its policy can be expressed exactly.
Also, the (regularized) best response of the second player can be computed in closed form using equation (\ref{eq:greedy}).
We update the first-moving player's policy using equation (\ref{mmd-update}) where $q_{\pi_t}$ is the feedback induced by the second-moving player's (regularized) PuB-AMG Nash equilibrium policy.

\subsubsection{Kuhn Poker}

We also investigate MiniMaxEnt objectives in an extensive-form game---Kuhn poker.
In Kuhn poker, there are up to three time steps.
For the third time step, we use the (regularized) PuB-AMG Nash equilibrium policy induced by equation (\ref{eq:greedy}).
For the first time step, at each iteration, we update the policy at each information state using MMD on the feedback from the previous time step.
For the second time step, at iteration $t$, holding fixed the iteration $t$ decision rule for the first time step, we use the policy induced by performing $\sqrt{t}$ iterations of MMD against the (regularized) PuB-AMG Nash equilibrium policy of the third time step.
As $\sqrt{t}$ grows large, we expect the decision rules for the second time step to approximate a PuB-AMG best response.

\begin{figure}
    \centering
    \includegraphics[width=\linewidth]{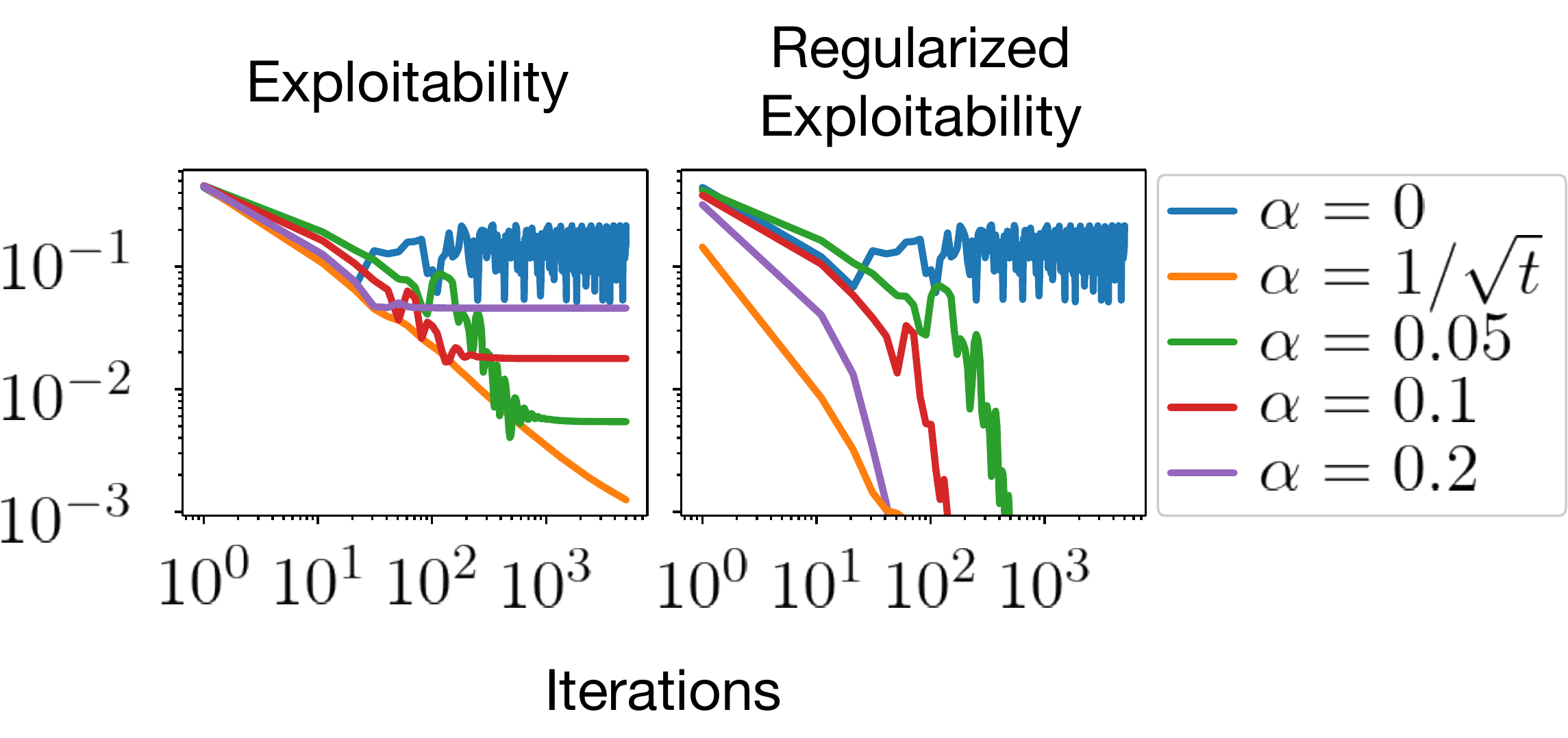}
    \caption{Results for Kuhn poker.}
    \label{fig:kuhn}
\end{figure}

We show the results for the original game in Figure \ref{fig:kuhn}.\footnote{We omit PuB-AMG (regularized) exploitability, as it is difficult to compute exactly in this case.}
Qualitatively, they are analogous to those from the perturbed rock-paper-scissors game.
The unregularized objective induces high exploitability iterates (blue) that do not converge in the original game; the objectives with fixed regularization (purple, red, greed) converge to constant exploitability and zero exploitability in the regularized game; the objective with annealed regularization converges to zero exploitability and zero regularized exploitability.

\section{Discussion on Reduction to Alternating Symmetric-Information Games} \label{sec:reduction}

In the the main body, we discussed a direct reduction from imperfect-information 2p0s games to PuB-AMGs.
In this section, we give a brief discussion on the intermediate reduction to alternating symmetric-information games, where symmetric information is meant as defined below.

\begin{definition}
A symmetric-information game is a game in which all players receive identical observations.
\end{definition}

This intermediate reduction is visualized in Figure \ref{fig:diagram2}.

\begin{figure}
    \centering
    \includegraphics[width=\linewidth]{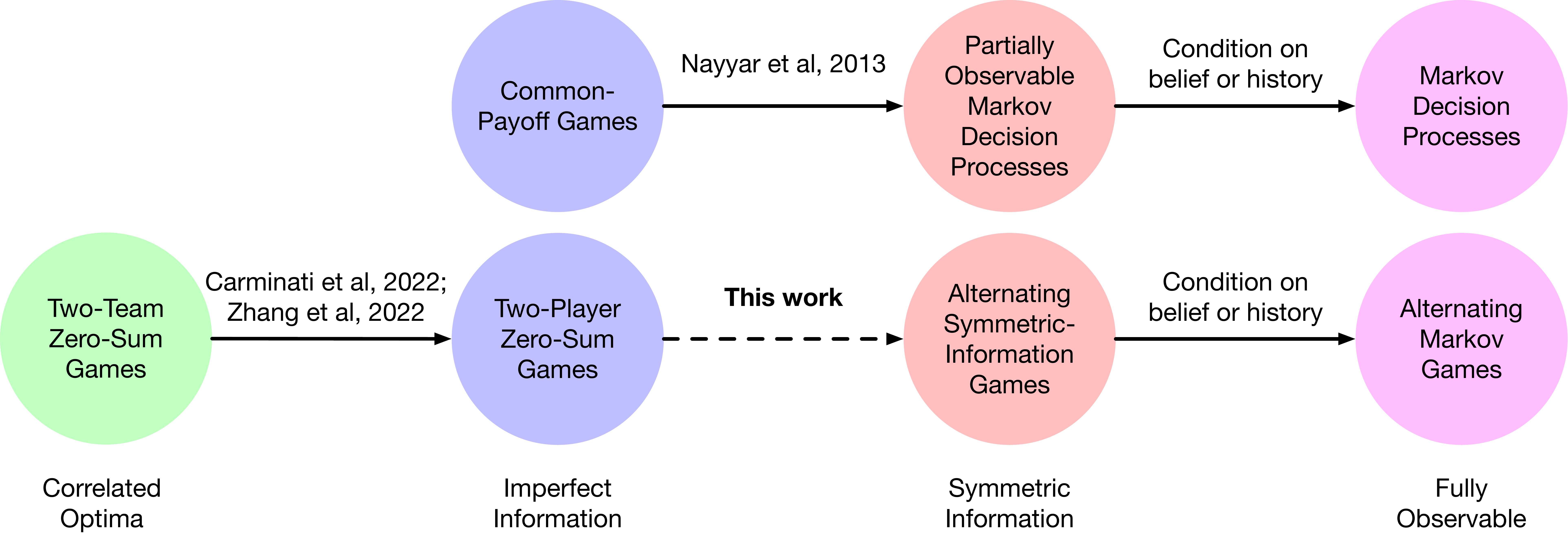}
    \caption{Our main contribution in the context of related work, at an abstract level. Solid lines denote reductions; the dashed line denotes a reduction that holds under a class of regularized objectives.}
    \label{fig:diagram2}
\end{figure}

\subsection{Finite-Horizon Symmetric-Information Sequential Games}

Symbolically, we say a setting is a finite-horizon symmetric-information sequential game if it can be described by a tuple
\[
\langle \mathbb{A}, \mathbb{O}, \mathbb{S}, \mathbb{H}, \mu,
\mathcal{O},
[\mathcal{R}_i], \mathcal{T}, T \rangle,
\] 
where 
\begin{itemize}[leftmargin=*]
    \item $i$ ranges from $0$ to $N-1$ and $\iota$ denotes the acting player.
    \item $\mathbb{A}$ is the set of actions. The actions of all players are assumed to be observable.
    \item $\mathbb{O}$ is the set of observations.
    \item $\mathbb{S}$ is the set of Markov states.
    \item $\mathbb{H} = \cup_t (\mathbb{O} \times \mathbb{A})^t \times \mathbb{O}$ is the set of histories. Because actions are observable and observations are identical, the history of the game is equal to each player's AOH.
    \item $\mu \in \Delta(\mathbb{S})$ is the initial state distribution.
    \item $\mathcal{O} \colon \mathbb{S} \to \mathbb{O}$ is the observation function.
    \item $\mathcal{R}_i \colon \mathbb{S} \times \mathbb{A} \to \mathbb{R}$ is the player $i$'s reward function.
    \item $\mathcal{T} \colon \mathbb{S} \times \mathbb{A} \to \Delta(\mathbb{S})$ is the transition function. 
    \item $T$ is the time horizon at which the game terminates.
\end{itemize}

\subsection{The Public Alternating Symmetric-Information Game}

Let
\[
\langle \mathbb{A},
[\mathbb{O}_i], \mathbb{O}_{\text{pub}}, [\mathbb{H}_i], \mathbb{H}_{\text{pub}},
\mathbb{H},
\mu,
[\mathcal{O}_i], \mathcal{O}_{\text{pub}},
[\mathcal{R}_i], \mathcal{T}, T \rangle,
\] 
be a finite-horizon 2p0s sequential game.
Then we define the associated public alternating symmetric-information game as the following finite-horizon fully-observable sequential game
\[
\langle \mathbb{\tilde{A}}, \mathbb{\tilde{O}}, \mathbb{\tilde{S}}, \mathbb{\tilde{H}}, \tilde{\mu},
\mathcal{\tilde{O}},
[\mathcal{\tilde{R}}_i], \mathcal{\tilde{T}}, \tilde{T} \rangle,
\] 
where
\begin{itemize}[leftmargin=*]
    \item $i$ ranges from $0$ to $1$ and $\iota  \in \{0, 1\}$ is the acting player.
    \item $\mathbb{\tilde{A}} = \{\tilde{a} \mid \tilde{a} \colon \mathbb{H}_{\iota}(h_{\text{pub}}) \to \Delta(\mathbb{A}), h_{\text{pub}} \in \mathbb{H}_{\text{pub}}\}$ is the set of \textit{public decision rules}.
    \item $\mathbb{\tilde{O}} = \mathbb{O}_{\text{pub}}$ is the set of observations.
    \item $\mathbb{\tilde{S}} = \mathbb{H}$ is the set of Markov states.
    \item $\mathbb{\tilde{H}} = \cup_t (\mathbb{\tilde{O}} \times \mathbb{\tilde{A}})^t \times \mathbb{\tilde{O}}$ is the set of histories.
    \item $\tilde{\mu} = \mu$ is the initial state distribution.
    \item $\mathcal{\tilde{O}} \colon \tilde{s} \mapsto \mathcal{O}_{\text{pub}}(\tilde{s})$ is the observation function.
    \item $\mathcal{\tilde{R}}_i \colon (\tilde{s}, \tilde{a}) \mapsto \mathbb{E}_{A \sim \tilde{a}(\tilde{s})} \mathcal{R}_i(\tilde{s}, A)$ is player $i$'s reward function.
    \item $\mathcal{\tilde{T}}(\tilde{s}' \mid \tilde{s}, \tilde{a}) = 
    \mathbb{E}_{A \sim \tilde{a}(\tilde{s})} \mathcal{T}(\tilde{s}' \mid \tilde{s}, A)$
    is the transition function.
    \item $\tilde{T} = T$.
\end{itemize}

As with POMDPs and MDPs, alternating symmetric-information games can be reduced to AMGs by either:
\begin{enumerate}
    \item Treating the (publicly-observable) history $\tilde{h} \in \mathbb{\tilde{H}}$ as the state.
    \item Treating the posterior $\mathcal{P}(\tilde{S} \mid \tilde{h})$ over the state (i.e., the history of the original game) given the (publicly observable) history as the state.
\end{enumerate}
The latter of these two conversions yields the PuB-AMG discussed in the main body.

\end{document}